\let \oldthanks \thanks
\documentclass[10pt,conference,compsocconf]{IEEEtran}
\usepackage{amsmath,amsfonts,amssymb,amsthm}
\usepackage{graphicx}
\usepackage{cite}
\usepackage{threeparttable}
\usepackage{booktabs}
\usepackage{mathrsfs}
\usepackage{multirow}
\usepackage{subfig}
\usepackage{bm}
\usepackage{float}
\usepackage{color}
\usepackage{makecell}
\usepackage{listings}
\usepackage[pass]{geometry}
\usepackage{enumitem}
\let \thanks \oldthanks

\newtheorem{theorem}{Theorem}
\newtheorem{lemma}{Lemma}

\newcommand\graphlabel{\vspace{-4pt}}%-6
\newcommand\graphtext{\vspace{-6pt}}%-10
\newcommand\graphgraph{\vspace{0pt}}%-4
\newcommand\graphbottomtext{\vspace{0pt}}%-10
\newcommand\tablelabel{\vspace{-4pt}}%-4
\newcommand\equtext{\vspace{-3pt}}%-4
\newcommand\titletext{\vspace{0pt}}%-12
\newcommand\authortext{\vspace{-8pt}}%-28

\begin{document}

\title{\LARGE{\textbf{Communication Lower Bound in Convolution Accelerators}}\titletext}

\author{\IEEEauthorblockN{Xiaoming Chen$^{1,2}$, Yinhe Han$^{1,2}$, Yu Wang$^{3}$}
\IEEEauthorblockA{$^1$Center for Intelligent Computing Systems, State Key Laboratory   of Computer Architecture, \\
Institute of Computing Technology, Chinese Academy of Sciences, Beijing, China\\
$^2$University of Chinese Academy of Sciences, Beijing, China\\
$^3$Department of Electronic Engineering, Tsinghua University, Beijing, China\\
Email: \{chenxiaoming, yinhes\}@ict.ac.cn, yu-wang@tsinghua.edu.cn}\authortext
\thanks{This paper will appear in 2020 26th IEEE International Symposium on High-Performance Computer Architecture (HPCA'20).}
}

%\author{
%Xiaoming Chen$^{1,2}$, Yinhe Han$^{1,2}$*, Yu Wang$^{3}$\\
%\affaddr{$^1$Center for Intelligent Computing Systems, State Key Laboratory   of Computer Architecture, }\\
%\affaddr{Institute of Computing Technology, Chinese Academy of Sciences, Beijing, China}\\
%\affaddr{$^2$University of Chinese Academy of Sciences, Beijing, China}\\
%%{$^3$Beijing Academy of Artificial Intelligence}\\
%\affaddr{$^3$Department of Electronic Engineering, Tsinghua University, Beijing, China}\\
%\affaddr{*Corresponding author}\\
%\email{\{chenxiaoming, yinhes\}@ict.ac.cn, yu-wang@tsinghua.edu.cn}\authortext
%}

\maketitle

%\title{Theoretically Minimizing Off-chip Communications for ASIC-based Convolution Accelerators}
%\title{A Communication-optimal Accelerator for Convolutional Neural Networks}

%\author{
%\IEEEauthorblockN{Xiaoming Chen, Yinhe Han}
%\IEEEauthorblockA{State Key Laboratory of Computer Architecture\\
%Institute of Computing Technology\\
%Chinese Academy of Sciences\\
%Email: \{chenxiaoming, yinhes\}@ict.ac.cn}
%\and
%\IEEEauthorblockN{Yu Wang}
%\IEEEauthorblockA{Department of Electronic Engineering\\
%Tsinghua University\\
%Email: yu-wang@tsinghua.edu.cn}
%}

\begin{abstract}
In current convolutional neural network (CNN) accelerators, communication (i.e., memory access) dominates the energy consumption. This work provides comprehensive analysis and methodologies to minimize the communication for CNN accelerators. For the off-chip communication, we derive the theoretical lower bound for any convolutional layer and propose a dataflow to reach the lower bound. This fundamental problem has never been solved by prior studies. The on-chip communication is minimized based on an elaborate workload and storage mapping scheme. We in addition design a communication-optimal CNN accelerator architecture. Evaluations based on the 65nm technology demonstrate that the proposed architecture nearly reaches the theoretical minimum  communication in a three-level memory hierarchy and it is computation dominant. The gap between the energy efficiency of our accelerator and the theoretical best value is only  37-87\%.
\end{abstract}

%\keywords{Convolutional neural network (CNN), CNN accelerator, communication lower bound}


\section{Introduction}\label{sec:intro}
Convolutional neural networks (CNNs) have achieved great successes in numerous practical applications (e.g., \cite{cnn_nips2012,cnn_cvpr2014,nlp_icml2008}). The reliable results produced by modern CNNs exclusively rely on the complex models and large amounts of data, which in turn bring significant demands in both performance and energy efficiency. Recently, a number of hardware accelerators based on either application-specific integrated circuits (ASICs) or field-programmable gate arrays (FPGAs) have been proposed to boost the performance and the energy efficiency of CNNs. %Various optimization techniques have also been proposed for CNN accelerators.

Due to the large amount of data and complex data reuse patterns in convolution computation, CNN accelerators often involve a great number of memory accesses. Inputs and weights are typically stored in the off-chip dynamic random-access memory (DRAM). A static random-access memory (SRAM) based on-chip global buffer (GBuf) stores portions of inputs and weights which are loaded from the DRAM. Each processing element (PE) has some registers  (Regs) to store inputs and weights which are read from the GBuf. Partial sums (Psums) are stored in the GBuf or Regs. During computation, there is complex data transmission in the memory hierarchy. Normally, communication, but not computation, dominates the  energy consumption of a CNN accelerator. A DRAM access consumes 2 to 3 orders of magnitude higher energy than an arithmetic operation~\cite{energy_isca2016} and the DRAM access energy can be  more than 90\% of the total energy consumption of a CNN accelerator~\cite{one_asplos2014,cx_micro2016}. For the on-chip aspects, Regs can take up a large portion ($>$50\%) of the chip energy  while  arithmetic units consume less than 20\%~\cite{one_isca2016}. %Hence, the energy efficiency of CNN accelerators is mainly limited by communication.
Therefore, from an energy point of view, %the computational energy is essential but the communication energy should be minimized.
current CNN accelerators are communication dominant. Minimizing the communication, therefore, is the key for  improving the energy efficiency of CNN accelerators.

Maximizing data reuse in convolutions helps reduce  communication. Data reuse heavily depends on the convolutional dataflow. There are various approaches to optimize the dataflow: 1) designing an elaborate dataflow~\cite{one_tnn2018,one_tcasi2018,one_isca2016,eyeriss_jssc2017,one_iccd2013,one_isca2015,one_asplos2014,one_asap2009,one_fpga2016,one_date2017,one_dsd2015,one2_tcasi2018,one_icpads2017}, 2) selecting the best dataflow from several candidates~
\cite{select_date2018,select_dac2016,select_tvlsi2017,select_hpca2017,one_aspdac19}, and 3) design space exploration (DSE)~\cite{exp_fpga2015,exp_iscas2017,exp_tvlsi2018,exp_arxiv2016,exp_date2015,exp_fpga2017,exp_aspdac2016,exp_dac2017,exp_fccm2016}. A fair number of these studies focus on the performance and/or the energy efficiency of the computational components. %Some of them (e.g., Eyeriss~\cite{eyeriss_jssc2017,one_isca2016} and Ref.~\cite{one_icpads2017}) focus on the on-chip dataflow.
The energy-dominant component, communication, has not been comprehensively investigated. Moreover, in most existing studies, the dataflow is designed based on intuitive/heuristic analysis, which may not guarantee the optimality.

%they have several drawbacks. The first two approaches fail to get the optimal solution due to the intuitive analysis and the limited candidates. The last one is well suited to FPGAs but is not good for ASICs, as it generates dedicated solutions for  given CNNs. More importantly, most of the existing CNN accelerators are designed for maximized performance and/or energy efficiency of the computational components, but the off-chip communication, which actually dominates the total energy consumption, has not been well considered or studied.

If the inputs, weights, and outputs of a convolutional layer are accessed exactly once at every level of the memory hierarchy, the layer-wise minimum communication is obviously reached. However, such an ambitious goal requires a huge on-chip memory. The requirement of memory resources varies for different applications. Thus, under given hardware resources, searching for a dataflow and an architecture that minimize the communication has much more practical significance. %It also helps CNN accelerators expand the scope of supported CNN models.
This problem has  never been solved. In this work, we provide detailed analysis and methodologies to reach the lower bounds of both off-chip communication and on-chip communication. Specifically, we make the following contributions in this paper.
%Specifically, we derive the theoretical lower bound of the off-chip communication for any convolutional layer that is implemented on a CNN accelerator with a given on-chip buffer capacity. %The optimality is guaranteed in the sense of achieving the theoretical lower bound of the complexity of the off-chip DRAM access volume.
%Based on the theoretical results, we design a. Since off-chip communication consumes much higher energy than arithmetic operations, our proposed dataflow and architecture improve the energy efficiency by ??? compared with the state-of-the-art CNN accelerators, XXX.
\begin{itemize}
\item We solve a fundamental problem in CNN accelerators: \textit{what the  lower bound of the off-chip communication of a convolutional layer is, if it is implemented on a CNN accelerator with a limited on-chip memory.} We provide a mathematical derivation for this problem. %The derivation relies on the red-blue pebble game~\cite{io_stoc1981}, which is a theoretical model used to derive the lower bound of the off-chip communication of algorithms, under a given on-chip memory capacity.
 \item We demonstrate that convolutions  have only one more level of data reuse (sliding window reuse) than matrix multiplications (MMs). Based on this conclusion, we elaborate a dataflow which fuses sliding window reuse and a communication-optimal  MM implementation,  to minimize the off-chip  communication.
%A novel dataflow that combines sliding window reuse and an optimal  MM implementation is proposed, to achieve  the minimum off-chip communication of any convolutional layer. The on-chip dataflow is also elaborated based on the same principle.
\item We propose a workload and storage mapping scheme such that both GBuf communication and Reg communication  respectively reach their  lower bounds.
%Our strategy greatly reduces the GBuf size.
\item A communication-optimal CNN accelerator architecture is proposed, which not only reaches the minimum communication, but also can adapt to different convolutional layer dimensions with high resource utilization.
\end{itemize}

The significance of this work is not purely on the proposed dataflow and/or architecture, but more importantly, from the point of view of a theoretical basis, to reveal the design methodology and principle to minimize the communication for CNN accelerators.

\section{Background}\label{sec:back}
%This section introduces the background of this work, including the basic concepts of convolutional layers and the related work about dataflow optimizations for CNN accelerators.

\subsection{Convolutional Layers}

Fig.~\ref{fig:conv} illustrates a general convolutional layer in CNNs. We have $B$ input images in a batch and $C_O$ kernels of weights, producing $B$ output images (only 1 input image and 1 output image are shown in Fig.~\ref{fig:conv}). Each input image has $C_I$ channels and each output image has $C_O$ channels. The output channel dimension is $H_O\!\times\! W_O$. Each  kernel is a $C_I\! \times\! H_K \!\times\! W_K$ 3D matrix. Each output is computed by an inner product between the inputs in a sliding window on the input image and the weights in a kernel. %Adjacent outputs correspond to different sliding windows and the position difference between two adjacent sliding windows is the stride size.
The stride size is the position difference between two adjacent sliding windows. Fig.~\ref{fig:code} lists the pseudo code of a convolutional layer. It contains 7 levels of loops and assumes the unit stride size.

\begin{figure}[t]
  \centering
  \includegraphics[width=1\columnwidth]{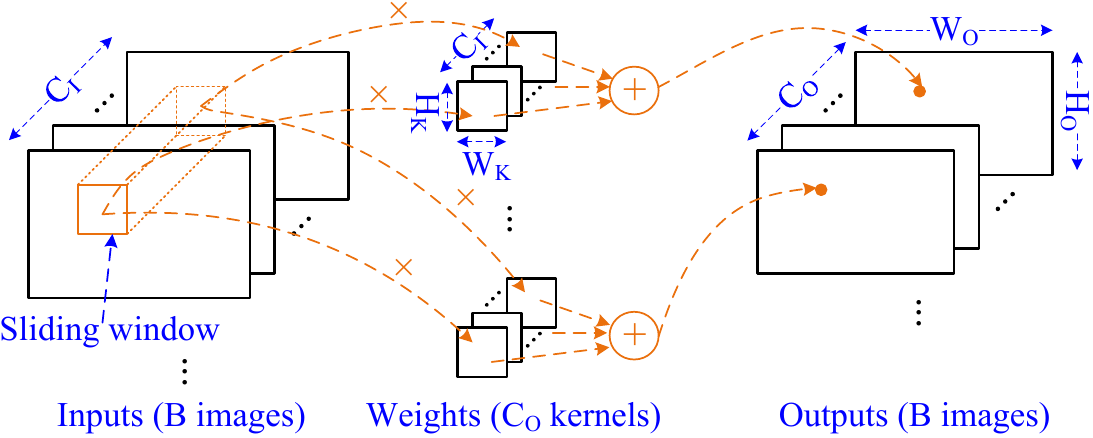}\\
  \graphlabel
  \caption{Convolutional layer in CNNs.}\label{fig:conv}\graphgraph
\end{figure}

\begin{figure}[t]
  \centering
  \includegraphics[width=.9\columnwidth]{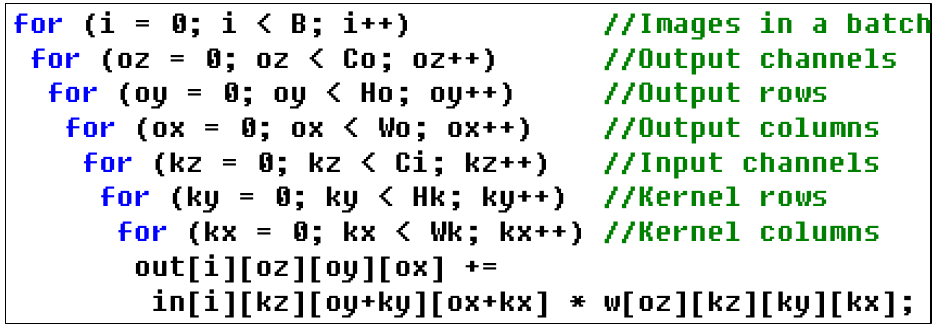}\\
  \graphlabel
  \caption{Pseudo code of a  convolutional layer.}\label{fig:code}\graphtext
\end{figure}

%\begin{verbatim}
% for (oz=0; oz<Co; oz++)
%  for (oy=0; ox<Ho; oy++)
%   for (ox=0; ox<Wo; ox++)
%    for (kz=0; kz<Ci; kz++)
%     for (ky=0; ky<Hk; ky++)
%      for (kx=0; kx<Wk; kx++)
%       out[oz][oy][ox] +=
%        in[kz][oy+ky][ox+kx]
%         * w[oz][kz][ky][kx];
%\end{verbatim}

%Without loss of generality, we consider the unit stride size in convolutions. But our methodology can also be applied to other stride sizes. Actually, convolutions with non-unit stride sizes can also be equivalently converted into convolutions with the unit stride size~\cite{one_tcasi2018}.

From a quick glance of Fig.~\ref{fig:code}, finding a dataflow with minimized  communication is challenging, due to the huge search space caused by different loop orders, loop stride sizes, loop unrolling schemes, etc. %Maximizing data reuse can minimize the off-chip communication and also boost the performance and the energy efficiency.
There are several data reuse patterns in convolutions, including input reuse (InR, an input is used by multiple kernels), sliding window reuse (WndR, an input is used by multiple overlapped sliding windows), weight reuse (WtR, a weight is used by multiple inputs), and output  reuse (OutR, an output  resides on chip during the entire computational process). Multiple data reuse patterns can be combined to form more complicated dataflows. Maximizing  data reuse also involves a huge search space.

In this work, we only consider the ordinary convolution algorithm, which is the most popular approach adopted by hardware accelerators. Those convolution algorithms with lower computational complexity, such as the Winograd algorithm~\cite{wino_cvpr2016} and fast Fourier transform based approaches~\cite{fft_2013}, are not considered. We target at minimizing the  communication of general convolution operations, so that our approach can be adopted in both inference and training of CNNs.

%Different dataflows result in different off-chip communication volumes, and thus, different performance, energy efficiency, etc. Many existing studies have tried to optimize the dataflow as well as the data reuse scheme and claimed their superiorities. However, we will show that none of them have obtained the optimal solution.

\subsection{Related Work}\label{sec:work}
%Here we review the related work about dataflow optimizations for CNN accelerators. The existing approaches are classified into three categories (designing an elaborate dataflow, selecting the best dataflow from multiple  candidates, and DSE), which are mentioned in Section~\ref{sec:intro}.

%There are several typical data reuse schemes in convolutions. Ref.~\cite{select_date2018} proposed to select the best data reuse scheme from three pre-defined schemes to achieve the minimum off-chip communication. However, the optimality cannot be guaranteed due to the few options. It is possible to get the lower bound of off-chip communication by exhaustively considering all possible loop permutations and tiling sizes for a given convolutional layer~\cite{block_arxiv2016,acc_fpga2017}. However, this approach faces a major limitation of the enormous search space. For example, only considering two parameters to optimize the off-chip communication can lead to a search space of 7.2$\times$10$^{13}$~\cite{acc_fpga2017}. Heuristic algorithms have to be used to find sub-optimal solutions. In addition, both approaches are suitable for only FPGA-based CNN accelerators. Due to the reconfigurability of FPGAs, we can derive the optimal data reuse scheme with any desired objective case by case, meaning that given a convolutional layer, it is able to get an ad hoc data reuse scheme with minimized off-chip communication.

A number of CNN accelerators are designed with an elaborate dataflow to optimize some objective(s) (e.g., performance, bandwidth, etc.)~\cite{one_tnn2018,one_tcasi2018,one_isca2016,one_iccd2013,one_isca2015,eyeriss_jssc2017,one_asplos2014,one_asap2009,one_fpga2016,one_date2017,one_dsd2015,one2_tcasi2018,one_icpads2017}. Unfortunately, their dataflows are designed almost based on intuitive/heuristic analysis. In other words, they  claimed the superiorities of the dataflows and/or accelerators but failed to explain why the designs are essentially the best. Such designs may not guarantee the optimality. A representative state-of-the-art is Eyeriss~\cite{one_isca2016,eyeriss_jssc2017} which claimed that the communication is optimized. We will show by experiments that neither its off-chip communication nor its on-chip communication is minimized.

Rather than using a single dataflow, several studies have integrated multiple dataflows into an accelerator (with increased hardware cost) and selected the best one according to the layer dimensions~\cite{select_date2018,select_dac2016,select_tvlsi2017,select_hpca2017,one_aspdac19}. These approaches usually perform better than the approaches based on a single dataflow. However, the claimed optimality is only the best one among the given candidates. If the defacto best solution is not included in the candidates, they cannot find the optimal solution.

%achieved by comparing the limited candidates so the optimal solution is

%In the existing CNN accelerators belonging to this category, each candidate dataflow only involves a single data reuse pattern, which is one of input reuse, sliding window reuse, weight reuse, output reuse, etc. We will show by both proofs and experiments that, none of these data reuse patterns can achieve the minimum off-chip communication. Instead, these data reuse patterns should be combined to take all of them into account at once.

To find the optimal dataflow with a particular objective, a possible approach is to exhaustively consider all possible loop orders and tiling sizes (i.e., the stride sizes for the loops). This is the DSE approach~\cite{exp_fpga2015,exp_iscas2017,exp_tvlsi2018,exp_arxiv2016,exp_date2015,exp_fpga2017,exp_aspdac2016,exp_dac2017,exp_fccm2016}. However, the search space is so huge that an exhaustive search is extremely time-consuming.  For instance, only considering two loops to minimize the off-chip communication of a particular layer leads to an enormous search space of 7.2$\times$10$^{13}$~\cite{exp_fpga2017}. Heuristics have to be adopted to find sub-optimal solutions. Exhaustive methods lack universality since they cannot tell people why the found dataflow is essentially the best. In this sense, for a new convolutional layer, re-conducting an exhaustive search is usually needed, as we do not know whether a known dataflow is still the best for the new layer.

{In the aforementioned approaches, some studies (\cite{one_iccd2013,one_isca2016,select_date2018,exp_fpga2017,exp_tvlsi2018,exp_date2015,exp_arxiv2016}) have considered communication optimization, while the others mainly focus on the computational components. Besides the three categories, there are other communication optimization approaches for CNN accelerators (e.g., the fused-layer approach~\cite{fuse_micro2016} that optimizes data movement between convolutional layers). Currently, no study has comprehensively analyzed the lower bound of communication in CNN accelerators.

%It computes several layers together, so that layers' outputs do not need to be shuffled on and off chip. However, its effectiveness is questionable. With the increased hardware overhead (the on-chip memory and logic are increased by 10-20\% and 50\%, respectively), the performance is  degraded by 7\% for VGGNet-E.}

Ref.~\cite{memreq_iiswc2018} analyzed the on-chip memory requirement such that both inputs and weights are read from the off-chip DRAM exactly once. This is the minimum possible off-chip communication. However, the required on-chip memory to achieve this goal is quite large (from several million bytes to hundreds of million bytes). On the other hand, the hardware resources are  fixed  but applications' requirements vary, so it is impossible to guarantee the goal all the time. In practice, searching for the minimum communication under given hardware resources has much more  significance.

\subsection{Preliminary: Red-blue Pebble Game}\label{sec:pebble}
Our derivation for the lower bound of the off-chip communication heavily depends on the red-blue pebble game~\cite{io_stoc1981}, which is a theoretical model to estimate the minimum volume of data transmission between two levels of memories. The derived lower bound is the best possible, in the sense that it is achievable by certain algorithm implementations. Here we review an important theorem of the red-blue pebble game as a preliminary of our derivation.

%In this section, we review the basic concepts of the red-blue pebble game~\cite{io_stoc1981}, as a necessary preliminary for the theoretical derivation in the next section.

%\subsection{Red-blue Pebble Game}\label{sec:game}
%The red-blue pebble game~\cite{io_stoc1981}

%In the red-blue pebble game, the memory hierarchy  is composed of a fast memory and a slow memory. The fast memory, which is considered to be the on-chip buffer/cache, can hold only $S$ data entries. The slow memory, which is considered to be the off-chip DRAM, is arbitrarily large. An algorithm is described by a directed acyclic graph (DAG), in which each node represents a data entry (and also an operation producing the data entry if the data entry is not an input), and each edge represents a data dependency.

Suppose that the memory hierarchy comprises \textit{an unlimited slow memory} and \textit{a limited fast memory}. When optimizing the off-chip communication, they refer to the off-chip DRAM and the on-chip memory (e.g., SRAMs or Regs), respectively. The fast memory can hold only $S$ data entries. An algorithm is described by a directed acyclic graph (DAG), in which each node represents a data entry or an operation (producing a data entry as the output of the operation) and each edge represents an inter-data dependency.  We skip the definition of the original red-blue pebble game here  because it is usually difficult to use~\cite{pebble_mit2017}. Instead, the red-blue pebble game can equivalently be converted to an easier $S$-partition problem~\cite{io_stoc1981}, which is defined as follows.

%The game is as follows. Given $S$ red pebbles and an arbitrary number of blue pebbles, with an initial blue pebble on each input node (representing that inputs are initially in the DRAM), a complete computation is any sequence of steps using the following four rules that results in a final state with a blue pebble on each output node (representing that outputs are written back to the DRAM).
%\begin{itemize}
%\item Rule 1: a red pebble may be placed on any node that has a blue pebble (loading a data entry from the off-chip DRAM to the on-chip memory).
%\item Rule 2: a blue pebble may be placed on any node that has a red pebble (storing a data entry from the the on-chip memory to the off-chip DRAM).
%\item Rule 3: if all the immediate predecessors of a node have red pebbles, a red pebble may be placed on that node (executing a ready operation and producing its output).
%\item Rule 4: a red pebble may be removed from any node (discarding a data entry in the on-chip memory).
%\end{itemize}

%Obviously, all possible data shuffling patterns are included in the red-blue pebble game. The off-chip communication volume to finish an algorithm is the number of the transitions caused by Rules 1 and 2 during the game. Directly  using the red-blue pebble game to derive the lower bound of the off-chip communication is usually difficult~\cite{pebble_mit2017}. Fortunately, the problem can equivalently be converted to an easier $S$-partition problem, which is defined as follows.

%\subsection{$S$-Partition}\label{sec:spar}

Let $G(V,E)$ be a DAG describing an algorithm, where $V$ and $E$ are  the node and edge sets, respectively. A partition on $G$ is called an \textit{$S$-partition}, if the following four properties hold.

%\begin{itemize}
\noindent $\bullet$ Property 1: $V$ is partitioned into $h$ subsets $V_1,\!V_2,\!\cdots,\!V_h$, such that $V_i$'s are disjoint and their union is $V$.

\noindent $\bullet$ Property 2: there is no cyclic dependency among $V_i$'s.

\noindent $\bullet$ Property 3: for any $V_i$ ($1\!\le\! i\! \le\! h$), there exists a dominator set $D_i$ (nodes in $D_i$ are not necessarily in $V_i$) such that $|D_i|\!\le\! S$. A dominator set $D_i$ for $V_i$ is a set of nodes in $V$ such that any path from an input of $G$ to a node in $V_i$ contains some nodes in $D_i$.

\noindent $\bullet$ Property 4: for any $V_i$ ($1\!\le\! i \!\le\! h$), the output set size $|O_i|\!\le\! S$. The output set $O_i$ of $V_i$ is a set of nodes in $V_i$ that do not have any immediate successors in $V_i$.
%\end{itemize}

Let $P(S)$ be the minimum number of subsets that any $S$-partition of a DAG must have.  The following theorem describes the communication lower bound based on the $S$-partition model (the proof is provided in~\cite{io_stoc1981}).
\begin{theorem}\label{theo:pebble}
Given a fast memory of size $S$, to finish a DAG that describes an algorithm, the minimum communication volume $Q$ between the fast memory and the slow memory satisfies\equtext
\begin{equation}\label{eq:comm}
Q \ge S \cdot \left( {P(2S) - 1} \right).
\end{equation}\equtext
\end{theorem}
%\noindent The proof of this theorem is provided in~\cite{io_stoc1981}. To calculate the lower bound of the off-chip communication for a given DAG, the primary task is to derive $P(2S)$, the minimum number of subsets in any $2S$-partition of the given DAG.

\section{Layer-Wise Lower Bound of Off-Chip Communication}\label{sec:derivation}

We now derive the layer-wise lower bound of the off-chip communication, based on the $S$-partition model~\cite{io_stoc1981} and the relation between convolutions and MMs. %We begin this section by showing the relation. Then we present the derivation and finally propose a dataflow that can achieve the lower bound.
Typically there are at least three levels in the memory hierarchy of a CNN accelerator: an off-chip DRAM, an on-chip GBuf, and Regs. The red-blue pebble game is still applicable. We  define an \textit{effective on-chip memory} as the maximum on-chip memory that does not contain duplicated data. For example, if the GBuf stores inputs and weights while some Regs store Psums (other Regs store inputs and weights that are copied from the GBuf), the effective on-chip memory refers to the GBuf (storing inputs and weights) plus those Regs which store Psums. A specific implementation may be a sub-optimum, since the red-blue pebble game assumes a homogeneous on-chip memory without any specific splitting.

\subsection{Relation Between Convolutions and MM}

\begin{figure}[t]
  \centering
  \includegraphics[width=.94\columnwidth]{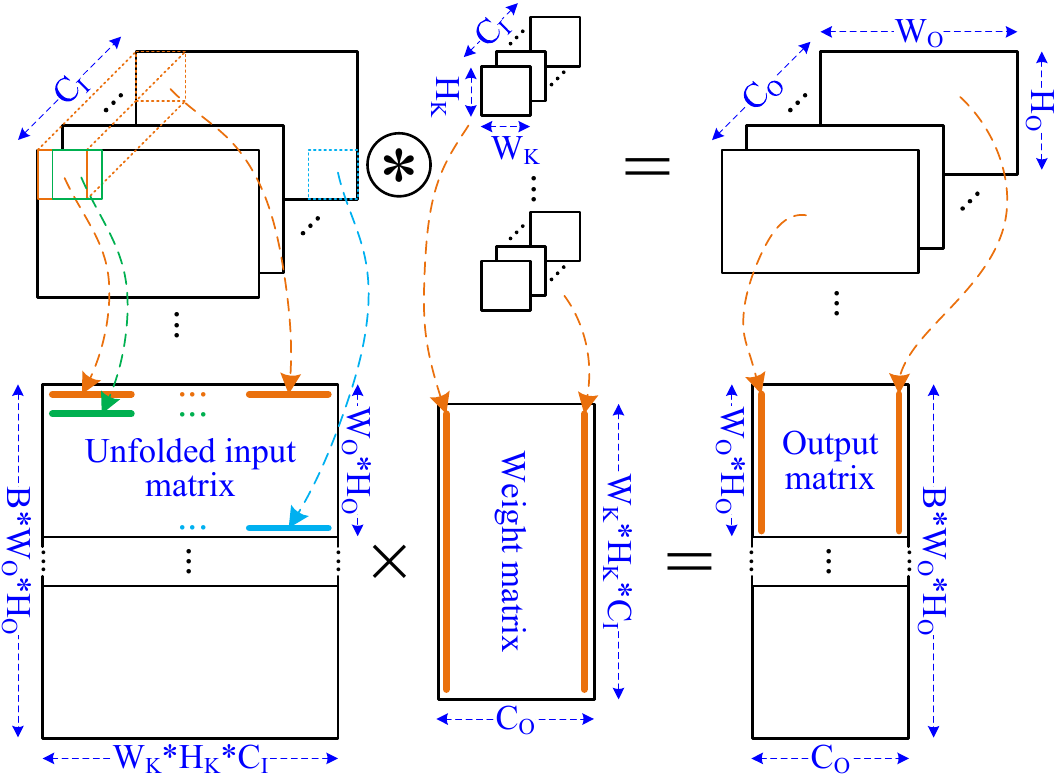}\\
  \graphlabel
  \caption{Convolution-to-MM conversion.}\label{fig:mm}\graphtext
\end{figure}

Fig.~\ref{fig:mm} shows how to convert a  convolutional layer into an MM. We first consider a simple case with batch size 1. The input image is \textit{unfolded} to the input matrix, each row of which contains the inputs in a sliding window. Different rows in the unfolded input matrix correspond to different sliding windows, which also correspond to different locations on the output channels. All  kernels are \textit{reshaped} into a weight matrix, each column of which contains the weights of a kernel. The output image is \textit{reshaped} into an output matrix, each column of which contains the outputs of an output channel. Reshaping means reorganizing the elements without adding or removing elements. If the batch size is $B$, we just stack up $B$ unfolded input matrices and $B$ output matrices, respectively, while the weight matrix remains unchanged. The stacked input and output matrices are still called unfolded input matrix and output matrix, respectively.

%By utilizing this method, any convolutional layer can be converted into an MM.
The convolution-to-MM conversion is only logic equivalent but not algorithm  equivalent. The difference is that, in a convolutional layer, inputs in overlapped sliding windows can be reused. This level of data reuse does not exist in MMs. This is why the input matrix is ``\textit{unfolded}'' instead of ``\textit{reshaped}''. In the conversion, the input images are unfolded by expanding all sliding windows, i.e., the common inputs in overlapped sliding windows have multiple explicit copies in the unfolded input matrix. We define $R$ to denote the reuse number of each input  by WndR, whose maximum value is\equtext
\begin{equation}\label{eq:reuse}
  R=\frac{W_K\times H_K}{D^2}\equtext
\end{equation}
where $D$ is the stride size. We will show that the derived  lower bound of  the off-chip communication relies on $R$.
%so no input activation can be reused by  sliding window reuse.

%For example, in Fig.~\ref{fig:mm}, the orange and green sliding windows are overlapped so those common pixels in both sliding windows can be reused in the two corresponding inner product operations. However, when they are unfolded, these common pixels have multiple explicit copies in the unfolded input matrix so no reuse can be exploited.

%Pixels on the boundaries of the input image may not be reused by so many times in the case of no padding. With proper padding sizes, pixels on the boundaries can also be reused by the same number.
%In practical implementations, $R$ may be smaller than $W_K\times H_K$, because fully reusing pixels needs much on-chip storage and a proper data shift pattern.
%
%\begin{figure}[t]
%  \centering
%  % Requires \usepackage{graphicx}
%  \includegraphics[width=.75\columnwidth]{reuse}\\
%  \caption{Each input activation can be reused by at most $W_K\times H_K$ sliding windows.}\label{fig:reuse}
%\end{figure}

One may argue that there are other data reuse patterns in a convolutional layer (e.g., InR, WtR, etc.). %For example, the kernels are shared by the computations of different output activations, and different kernels can share the same input activations.
These data reuse patterns are actually included in the converted MM. For example, each column of the weight matrix can be shared by multiple  rows in the unfolded input matrix, which is WtR, and each row of the input matrix can be reused by multiple  columns in the weight matrix, which is InR. From the conversion process, it is clear that the computational process of a convolutional layer is not changed except for WndR, because except for that the input matrix is unfolded, the other matrices are just reshaped. This  implies that, although a convolutional layer involves 7 levels of loops, it only has one more level of data reuse than MMs. In order to take into account WndR in the converted MM, we have defined $R$ to denote the reuse number of each input  by WndR. If $R$ is 1 (i.e., no WndR), %indicating that there is no data reuse by the sliding window reuse pattern,
a convolutional layer is exactly equivalent to an MM. Since a fully-connected (FC) layer is also  equivalent to an MM, our conclusion with $R\!=\!1$ can be applied to FC layers. Note that the convolution-to-MM conversion is a only a logical operation used for our derivation. It is not a real operation in our dataflow or architecture.

\subsection{Theoretical Derivation}

Here we provide the theoretical derivation for the layer-wise lower bound of the off-chip communication. We consider a general case in which the on-chip memory cannot hold all  inputs or all  weights of a convolutional layer. Otherwise, it is just the ideal case (both the inputs and the weights are read exactly once).

%\begin{table}[t]
%  \centering
%  \caption{}\label{tab:sym}
%  \begin{tabular}{ll}
%  \toprule
%  {}&{}\\
%  \midrule
%  {$C_I$}&{\# of input channels}\\
%  {$C_O$}&{\# of output channels}\\
%  {$W_O$}&{Width of output images}\\
%  {$H_O$}&{Height of output images}\\
%  {$W_K$}&{Width of kernels}\\
%  {$H_K$}&{Height of kernels}\\
%  \bottomrule
%  \end{tabular}
%\end{table}

\begin{lemma}\label{lem:tnodes}
If a convolutional layer is represented by a DAG, the number of internal and output nodes in the DAG is $2BW_OH_OC_OW_KH_KC_I$.
\end{lemma}

\begin{figure}[t]
  \centering
  \includegraphics[width=.94\columnwidth]{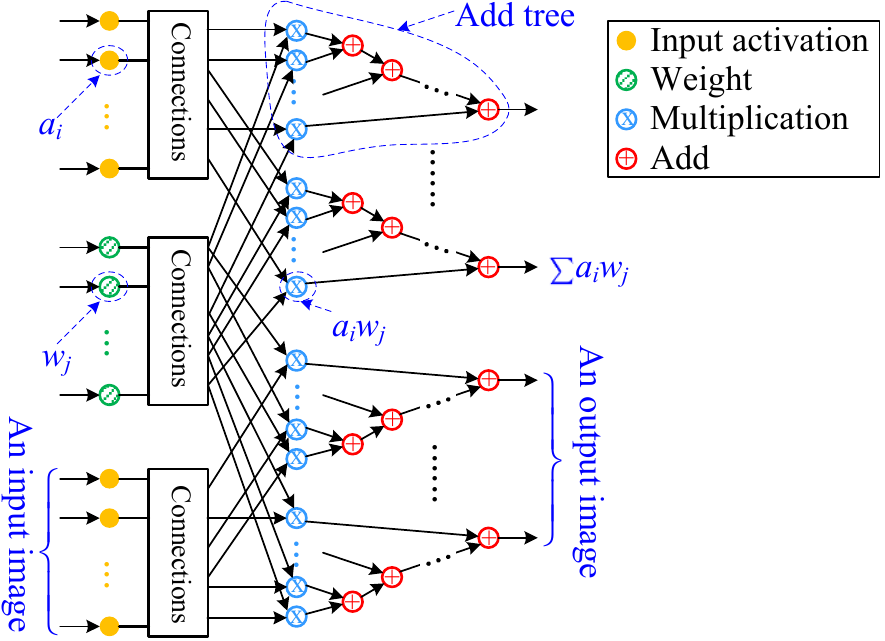}\\
  \graphlabel
  \caption{DAG describing a convolutional layer.}\label{fig:dag}\graphgraph
\end{figure}

\begin{proof}
Fig.~\ref{fig:dag} illustrates a DAG that describes a  convolutional layer. It has three levels. The first level comprises all input nodes, including inputs and weights. The second level is composed of all multiplication nodes. The last level is composed of all add nodes. The multiplication and add nodes associated with the same output form an \textit{add tree} (multiplication nodes are also included in add trees). The detailed connections between the input nodes and the multiplication nodes are not shown because we are not interested in them.

There are $BW_OH_OC_O$ outputs in total. Each output  is the inner product of two vectors, both of which are of length $W_KH_KC_I$. Hence, there are $W_KH_KC_I$ multiplications nodes and $W_KH_KC_I$ add nodes in an add tree. Since no internal node can be shared by different add trees, the number of internal and output nodes is $2BW_OH_OC_OW_KH_KC_I$.
\end{proof}

In Fig.~\ref{fig:dag}, the inputs are marked as $a_1,a_2,\cdots$, and the weights are marked as $w_1,w_2,\cdots$. Each multiplication node produces a \textit{term} $a_iw_j$. Note that the sum of multiple terms (e.g., $a_1w_1\!+\!a_2w_2$) is not called a term. Instead, a term belongs to a sum (i.e., an add tree). We have the following lemma.

\begin{lemma}\label{lem:term}
Let $T(S)$ be the maximum number of terms that can be produced  in no more than $S$ add trees by using no more than $S$ on-chip memory units. For a convolutional layer with each input reused by $R$ times by WndR, $T(S)\!=\!O(S\sqrt{RS})$.
\end{lemma}
\begin{proof}
The proof is based on the relation between convolutions and MMs. We use $\bf A$, $\bf B$ and $\bf C$ to denote the unfolded input matrix, the weight matrix, and the output matrix, respectively. Then the MM is represented as ${\bf AB}\!=\!\bf C$. The produced terms using no more that $S$ on-chip memory units can be arbitrarily distributed in $\bf C$. Note that an element in $\bf C$  is the sum of multiple terms belonging to the said sum, so the produced terms may overlap in $\bf C$.

We first demonstrate that in order to maximize the produced terms, the produced terms  must form a single block or be able to form a single block. This phenomenon can be explained intuitively. We consider any two elements (i.e., two sums) in $\bf C$, each of which is a product of a row vector in $\bf A$ and a column vector in $\bf B$. Obviously, if we move one element  such that the two elements are overlapped in $\bf C$ (then the input vectors are also overlapped), the number of produced terms keeps unchanged but the number of required on-chip memory units is minimized. In what follows, we provide a mathematical proof for this statement.

\begin{figure}[t]
  \centering
  \includegraphics[width=.6\columnwidth]{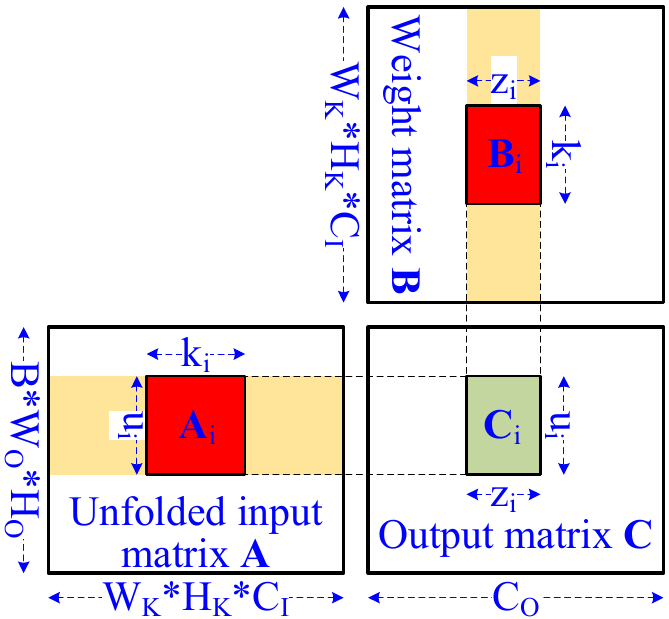}\\
 \graphlabel
  \caption{Converted matrix multiplication.}\label{fig:term}\graphtext
\end{figure}

Without loss of generality, we consider any two non-overlapped rectangular sum blocks ${\bf C}_1$ and ${\bf C}_2$ in $\bf C$. %suppose that the produced sums form two non-overlapped rectangular blocks ${\bf C}_1$ and ${\bf C}_2$ in $\bf C$. %${\bf C}_1$ and ${\bf C}_2$ cannot overlap (if they overlap, we can always re-partition the blocks to eliminate overlap).
The size of  ${\bf C}_i$ is $u_i\!\times\! z_i$ (the minimum  size is $1\!\times\! 1$). Block ${\bf C}_i$ is the product of two corresponding blocks in $\bf A$ and $\bf B$, respectively, say ${\bf A}_i$ and ${\bf B}_i$, which are of sizes $u_i\!\times \!k_i$ and $k_i\!\times\! z_i$, as shown in Fig.~\ref{fig:term}. Then each element in ${\bf C}_i$ is the sum of $k_i$ terms. Note that ${\bf A}_1$ and ${\bf A}_2$  can overlap, or ${\bf B}_1$ and ${\bf B}_2$  can overlap, but they cannot overlap at the same time (otherwise ${\bf C}_1$ and ${\bf C}_2$ will overlap). %, though different ${\bf C}_i$ blocks cannot overlap.
 According to the definition of $T(S)$, we intend to maximize\equtext
\begin{equation}\label{eq:term}
T(S)=u_1k_1z_1+u_2k_2z_2.\equtext
\end{equation}

If there is no overlap in any ${\bf A}_i$ or ${\bf B}_i$ blocks, since all inputs and outputs should be in no more than $S$ on-chip memory units, we have the following constraint\equtext
\begin{equation}\label{eq:cons1}
  \frac{{{u_1}{k_1}}}{R} + \frac{{{u_2}{k_2}}}{R} + {z_1}{k_1} + {z_2}{k_2} + {u_1}{z_1} + {u_2}{z_2} \le S\equtext
\end{equation}
where the first two terms on the left-hand side are reduced by a factor $R$ because an element in $\bf A$ can be reused by at most $R$ times by WndR. Equation~\eqref{eq:cons1} also implies that the produced terms are in no more than $S$ add trees.
%By the definition of $T(S)$, all  produced terms must be in no more than $S$ add trees (namely, $S$ sums), i.e.,\equtext
%\begin{equation}\label{eq:cons1}
%\sum\limits_{i = 1}^n {{u_i}{z_i}}  \le S.\equtext
%\end{equation}
%The elements in all ${\bf A}_i$ and ${\bf B}_i$ blocks must be in the on-chip memory, so the number of these elements cannot exceed $S$. However, since an element in $\bf A$ can be reused by at most $R$ times by WndR, the actual number of required on-chip memory units for elements in $\bf A$ is reduced by a factor of $R$. This means that\equtext
%\begin{equation}\label{eq:cons2}
%\sum\limits_{i = 1}^n {\left( {\frac{{{u_i}{k_i}}}{R} + {k_i}{z_i}} \right)}  - |OV| \le S\equtext
%\end{equation}
%where set $OV$ contains all overlapped elements in all ${\bf A}_i$ and ${\bf B}_i$ blocks. Clearly, $|OV|\!\le\! S$. Hence, \eqref{eq:cons2} becomes\equtext
%\begin{equation}\label{eq:cons3}
%\sum\limits_{i = 1}^n {\left( {\frac{{{u_i}{k_i}}}{R} + {k_i}{z_i}} \right)}   \le 2S.\equtext
%\end{equation}
Based on the generalized mean inequality~\cite{mean_url}, we have\equtext
\begin{equation}\label{eq:d1}
  \frac{{{u_i}{k_i}}}{R} + {z_i}{k_i} + {u_i}{z_i} \ge \frac{3}{{\sqrt[3]{R}}}{\left( {{u_i}{k_i}{z_i}} \right)^{\frac{2}{3}}}\equtext
\end{equation}
where the equality holds iff $\frac{{{u_i}{k_i}}}{R}\! =\! {z_i}{k_i} \!=\! {u_i}{z_i}$. Combining \eqref{eq:cons1} and \eqref{eq:d1}, we have\equtext
\begin{equation}\label{eq:d2}
  {\left( {{u_1}{k_1}{z_1}} \right)^{\frac{2}{3}}} + {\left( {{u_2}{k_2}{z_2}} \right)^{\frac{2}{3}}} \le \frac{{S\sqrt[3]{R}}}{3}.\equtext
\end{equation}
Let $t_i\!\buildrel \Delta \over = \!\left(u_ik_iz_i\right)^{\frac{2}{3}}$, then we have formulated a maximum value problem to maximize $T(S)\!=\!t_1^{\frac{3}{2}} + t_2^{\frac{3}{2}}$ under the constraint $t_1\!+\!t_2 \!\le\! \frac{{S\sqrt[3]{R}}}{3}$. Since $T(S)\!=\!t_1^{\frac{3}{2}} \!+\! t_2^{\frac{3}{2}}$ is continuous and  strictly convex, its upper bound is reached on the boundary of the variables' value range, i.e., when there is one $t_i\!=\!\frac{{S\sqrt[3]{R}}}{3}$ and the other $t_i$ is 0. Accordingly, the upper bound of $T(S)$ is $\frac{S\sqrt{RS}}{3\sqrt{3}}\!=\!O(S\sqrt{RS})$. The upper bound can be reached iff there is only one $i$ such that $z_i\!=\!\frac{\sqrt{S}}{\sqrt{3R}}$ and $u_i\!=\!k_i\!=\!\frac{\sqrt{SR}}{\sqrt{3}}$, implying that there is only a single block in $\bf C$.

If there is overlap between ${\bf A}_1$ and ${\bf A}_2$ or between ${\bf B}_1$ and ${\bf B}_2$, without loss of generality, we assume that ${\bf A}_1$ and ${\bf A}_2$  are overlapped (then ${\bf B}_1$ and ${\bf B}_2$ cannot overlap). The constraint is as follows\equtext
\begin{equation}\label{eq:cons2}
  \begin{aligned}
\frac{{{u_1}{k_1}}}{R} + {z_1}{k_1} + {z_2}{k_2} + {u_1}{z_1} + {u_2}{z_2} &\le S,\\
\frac{{{u_2}{k_2}}}{R} + {z_1}{k_1} + {z_2}{k_2} + {u_1}{z_1} + {u_2}{z_2} &\le S.
\end{aligned}\equtext
\end{equation}
This is equivalent to%\equtext
\begin{equation}\label{eq:d3}
  \begin{aligned}
{u_1}{k_1} &\le R\left( {S - {z_1}{k_1} - {z_2}{k_2} - {u_1}{z_1} - {u_2}{z_2}} \right),\\
{u_2}{k_2} &\le R\left( {S - {z_1}{k_1} - {z_2}{k_2} - {u_1}{z_1} - {u_2}{z_2}} \right).
\end{aligned}\equtext
\end{equation}
Then\equtext
\begin{equation}\label{eq:d4}
 \begin{aligned}
&T(S) = \sqrt {{u_1}{k_1}} \sqrt {{u_1}{k_1}} {z_1} + \sqrt {{u_2}{k_2}} \sqrt {{u_2}{k_2}} {z_2}\\
 &\le\! \!\sqrt {R\left( {S \!\!-\!\! {z_1}{k_1}\!\! -\!\! {z_2}{k_2} \!\!-\!\! {u_1}{z_1} \!\!-\!\! {u_2}{z_2}} \right)} \!\left( \!{\sqrt {{u_1}{k_1}} {z_1}\! \!+\!\! \sqrt {{u_2}{k_2}} {z_2}}\! \right)\!.
\end{aligned}\equtext
\end{equation}
Since ${z_i}{k_i} \!+\! {u_i}{z_i} \!\ge\! 2\sqrt {{k_i}{u_i}} {z_i}$, we have\equtext
\begin{equation}\label{eq:d5}
\begin{aligned}
  &T(S) \\
  &\le\!\! \sqrt {R\left( {S\!\! -\!\! 2\sqrt {{k_1}{u_1}} {z_1} \!\!-\!\! 2\sqrt {{k_2}{u_2}} {z_2}} \right)} \!\left(\! {\sqrt {{k_1}{u_1}} {z_1}\!\! +\!\! \sqrt {{k_2}{u_2}} {z_2}} \!\right)
  \end{aligned}\equtext
\end{equation}
where the equality holds iff $k_i\!=\!u_i$. Let $t\!\buildrel \Delta \over = \!\sqrt {{k_1}{u_1}} {z_1} +$ $ \sqrt {{k_2}{u_2}} {z_2}$. Based on the generalized mean inequality, we get\equtext
\begin{equation}\label{eq:d6}
  \begin{aligned}
&T(S) \le \sqrt {R(S - 2t)} \,t = \sqrt {R(S - 2t)} \sqrt t \sqrt t \\
& \le \sqrt R {\left( {\frac{{S - 2t + t + t}}{3}} \right)^{\frac{3}{2}}} = \frac{{S\sqrt {RS} }}{{3\sqrt 3 }} = O(S\sqrt {RS} )
\end{aligned}\equtext
\end{equation}
where the last equality holds iff $t\!=\!\frac{S}{3}$. When tracing back the derivation process, the upper bound of $T(S)$ can be reached iff ${k_1} \!= \!{k_2}\! =\! {u_1}\! =\! {u_2}\! =\! \frac{{\sqrt {SR} }}{{\sqrt 3 }}$ and ${z_1} \!+ \!{z_2} \!=\! \frac{{\sqrt S }}{{\sqrt {3R} }}$, where the latter condition implies that ${\bf C}_1$ and ${\bf C}_2$ are able to be merged into a single block, and the resulting case is identical to the case without overlap.

We have proved that for any two  blocks in $\bf C$, only when they form a single block or be able to form a single block, the number of produced terms is maximized. When extending this conclusion to the general case with multiple blocks in $\bf C$, they should also form a single block or be able to form a single block (say, ${\bf C}_1$) with $u_1\!=\!Rz_1$ held. If so, the upper bound of $T(S)$, $\frac{{S\sqrt {RS} }}{{3\sqrt 3 }}\! =\! O(S\sqrt {RS} )$, can be reached.
\end{proof}

%when there is only one ${\bf C}_i$ block in $\bf C$ (and accordingly, only one ${\bf A}_i$ block in $\bf A$ and only one ${\bf B}_i$ block in $\bf B$),

%$\bf A$ is partitioned into two parts, ${\bf A}_1$  and  ${\bf A}_2$, by the following method. ${\bf A}_1$ consists of those rows in $\bf A$, each of which has at least $\sqrt{S}$ entries belonging to $U$. ${\bf A}_2$ consists of the rest of rows in $\bf A$. Accordingly, $\bf C$ is also partitioned into two parts, ${\bf A}_1 \bf B$ and ${\bf A}_2 \bf B$. For convenience, we assume that all rows in ${\bf A}_1$ are the most front rows in $\bf A$.

%Let $U$ be any set of entries in $\bf A$ and $\bf B$, with $|U|\le S$.

%Since $|U|\le S$, ${\bf A}_1$ can have at most $\sqrt{S}$ rows.

\begin{lemma}\label{lem:inodes}
Let $\{V_1,V_2,\!\cdots\!,V_h\}$ be an $S$-partition of the DAG associated with a convolutional layer. Each $V_i$ ($1\!\le\! i\! \le \!h$) can have at most $2T(S)\!+\!S$ internal and output nodes.
\end{lemma}
\begin{proof}
By Property 4 of the $S$-partition model, the output set of $V_i$ has at most $S$ nodes. This implies that $V_i$ can have nodes in at most $S$ add trees. To bound the internal and output nodes that $V_i$ can have, we only need to consider  $S$ add trees. By property 3 of the $S$-partition model, there is a dominator set $D_i$ for $V_i$ that has no more than $S$ nodes. By the definition of $T(S)$, from $D_i$ at most $T(S)$ terms can be formed in $S$ add trees. $T(S)$ terms can form at most $T(S)$ add nodes in $V_i$. Considering that nodes in $D_i$ ($|D_i|\!\le\! S$) can possibly be internal or output nodes of $V_i$, $V_i$ can have at most $2T(S)\!+\!S$ internal and output nodes.
\end{proof}

Based on Lemmas~\ref{lem:tnodes}, \ref{lem:term} and \ref{lem:inodes}, for a DAG that describes a convolutional layer, the minimum number of subsets that any $S$-partition must have is\equtext
\begin{equation}\label{eq:spart}
  P(S)=\Omega \left( \frac{B{W_O}{H_O}{C_O}{W_K}{H_K}{C_I}}{S\sqrt{RS}} \right).\equtext
\end{equation}
According  to Theorem~\ref{theo:pebble}, we get the following theorem, which is also the key conclusion of this paper.
%about the lower bound of the off-chip communication of a general convolutional layer in CNNs.
\begin{theorem}\label{theo:conv}
The lower bound of the off-chip communication of a convolutional layer is\equtext
\begin{equation}\label{eq:conv}
  Q_{{\rm DRAM}}=\Omega \left( {\frac{{B{W_O}{H_O}{C_O}{W_K}{H_K}{C_I}}}{{\sqrt {RS} }}} \right).\equtext
\end{equation}
\end{theorem}

The off-chip communication volume of a naive convolution implementation (without any data reuse)  is simply $2BW_OH_OC_OW_KH_KC_I$. The lower bound reduces it by a factor of $\sqrt{RS}$. %As can be seen, the lower bound of the off-chip communication depends on $R$. %If $R$  achieves the theoretical maximum value $W_KH_K$, then
%\begin{equation}\label{eq:conv2}
%  Q_{\rm min} = \Omega \left( {\frac{{B{W_O}{H_O}{C_O}\sqrt {{W_K}{H_K}} {C_I}}}{{\sqrt S }}} \right).
%\end{equation}
If $R$ is 1, then a convolutional layer is exactly equivalent to an MM. In this case, the reduction factor is $\sqrt{S}$, which is consistent with the communication-optimal implementation of MMs~\cite{io_stoc1981}.

It is worth mentioning that the derived lower bound is in the form of $\Omega$ instead of a precise value. It represents the asymptotic relation between the off-chip communication  and the on-chip memory capacity when the problem scale is large enough. It is possible that some  dataflows can bring less off-chip communication in some special cases (e.g., cases of small workloads).

\section{Communication-Optimal Dataflow}\label{sec:dataflow}

In this section, we elaborate our dataflow with minimized off-chip communication based on the above derivation. The on-chip communication is minimized based on a proposed workload and storage mapping scheme.

\subsection{Dataflow with Minimized Off-Chip Communication}

The dataflow with minimized off-chip communication is derived from the proof process of Lemma~\ref{lem:term}. More precisely, in Fig.~\ref{fig:term}, the output matrix $\bf C$ is partitioned into equal-sized blocks of size $u\!\times\! z$. The block size should satisfy $u\!\approx\!Rz$ and also meet the on-chip memory capacity. %A block in $\bf C$ should be computed by the available hardware resources at a time.
A block needs the data in the two yellow bands in the unfolded input matrix $\bf A$ and the weight matrix $\bf B$. Actually, %if we  consider an ordinary MM (equivalent to a convolutional layer without WndR),
the communication-optimal implementation of MM is also the blocked method described in Fig.~\ref{fig:term}~\cite{gemm_toms2008}. When we map the blocked implementation back to a convolutional layer, we get Fig.~\ref{fig:imple}.

\begin{figure}[b]
\graphbottomtext
  \centering
  \includegraphics[width=1\columnwidth]{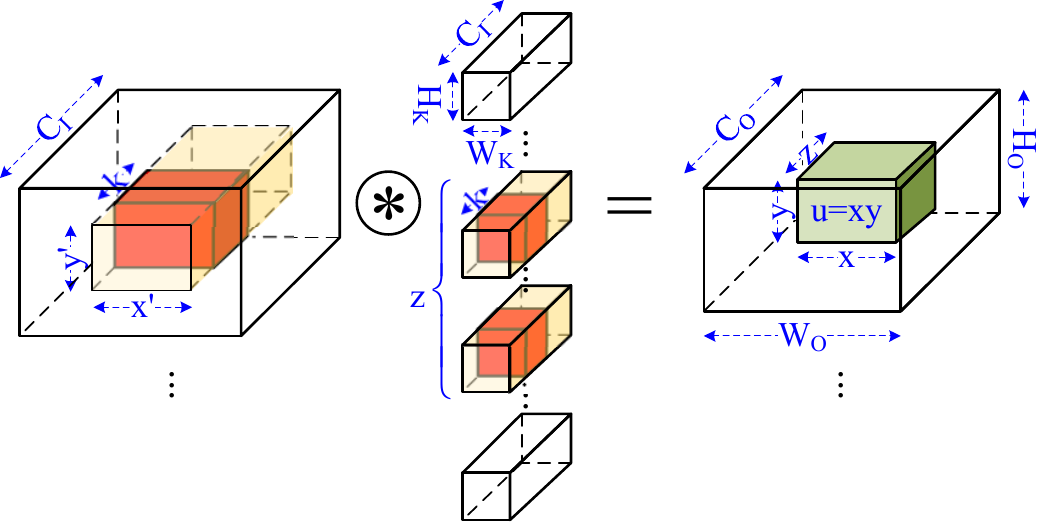}\\
 \graphlabel
  \caption{Dataflow to achieve the lower bound of off-chip communication.}\label{fig:imple}
\end{figure}

A block in $\bf C$ can be mapped to a $z\!\times\! y\times \!x$ ($u\!=\!xy$) 3D sub-matrix in the output images (e.g., the green block in Fig.~\ref{fig:imple}). %The sub-matrix consists of a portion of the locations on the output channels and a portion of the output channels.
If the output channel dimension is too small (i.e., $W_OH_O\!<\!xy$), the said output sub-matrix may be from multiple (say, $b$) images  in a batch. In this case, $u\!=\!bxy$. To compute the output sub-matrix, the inputs in the corresponding $x'\!\times\! y'$ ($x'\!=\!x\!+\!W_K\!-\!1$ and $y'\!=\!y\!+\!H_K\!-\!1$ if $D\!=\!1$) locations from all input channels of $b$ images (i.e., the yellow block in the input images) and $z$ kernels associated with the partial output channels (i.e., the  kernels colored yellow) are needed, as shown in Fig.~\ref{fig:imple}. Due to the limited on-chip memory, it might be impossible to load all required data at a time. Instead, it is computed by a series of \textit{iterations}.  In each {iteration}, in the yellow blocks, we load the inputs from a portion (say, $k$) of the input channels and the corresponding weights to the on-chip memory, shown by the red blocks in Fig.~\ref{fig:imple}. Then we can perform a partial update to the output sub-matrix.  To complete the output sub-matrix, we continuously load inputs and weights in the yellow blocks and perform partial updates. %, until all the data in the yellow blocks have been loaded exactly once.
For an output sub-matrix, the needed inputs and weights are read from the off-chip DRAM exactly once. Different output sub-matrices in the output images are computed sequentially in the same way. Fig.~\ref{fig:dataflow} lists the pseudo code of the dataflow. Any quadruple  $\{b,z,y,x\}$ (i.e., tiling sizes) defines an implementation of the dataflow. For a fixed quadruple $\{b,z,y,x\}$, $k$ does not affect the off-chip communication. However,  under a given on-chip memory capacity, smaller $k$ results in larger output sub-matrices, and thus, less output sub-matrices. Hence, $k$ should be  the smallest value, namely, 1.

%To maximize the output sub-matrix to reduce off-chip communication,

\begin{figure}[t]
  \centering
  \includegraphics[width=1\columnwidth]{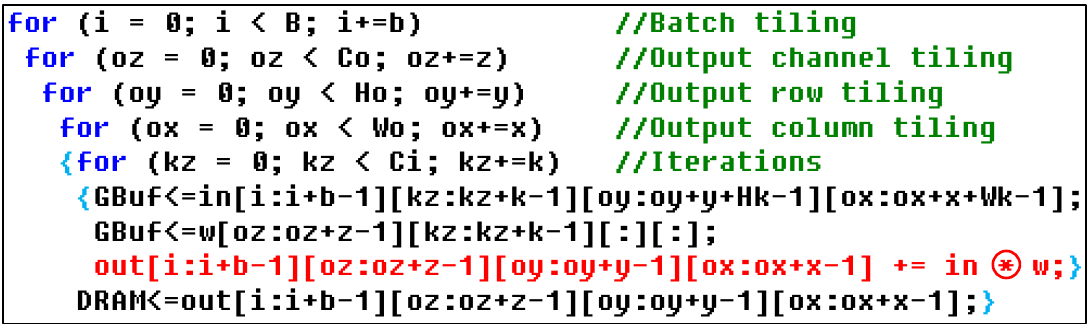}\\
  \graphlabel
  \caption{Pseudo code of the proposed dataflow.}\label{fig:dataflow}\graphtext
\end{figure}

To explain why this dataflow is superior, we notice that it fully exploits OutR, since Psums reside on chip during the computational process and are written back to the off-chip DRAM only once after the computation is finished. WndR (an input is reused by at most $R$ sliding windows on each $x'\!\times \!y'$ plane) is also fully exploited. More importantly, it also takes into account  InR (an input is reused by weights in $z$ kernels) and WtR (a weight is reused by $b\!\times\! x\!\times\! y$ inputs)  at the same time. However, neither InR nor WtR is fully utilized (for example, the loaded inputs are only reused by the loaded weights but not by all kernel weights). This implies that maximizing either InR or WtR is never the optimal solution. In fact, our approach utilizes InR and WtR in a balanced way, generating equal loading volumes of inputs and weights. To sum it up, our dataflow fully exploits OutR and WndR and also combines InR and WtR in a balanced way.

%Note that our derivation can simply be extended to sparse CNNs. The rectified linear unit (ReLU) function is widely adopted in modern CNNs, which lead to many zeros in the activations. Avoiding accessing zero activations also helps reduce the off-chip communication. If the sparsity of the input activations of a convolutional layer is $\gamma$ ($\le 1$), it is equivalent to use $R/\gamma$ to replace $R$ in the derivation. Except for this difference, our conclusion remains the same.

We now verify that the proposed dataflow is able to achieve the lower bound of the off-chip communication.  There are ${(BW_OH_OC_O)}\!/\!{(bxyz)}$ blocks in total in the output images. For each block, $W_KH_KC_Iz$ weights and $bx'y'C_I$ inputs are needed. The DRAM read volume is\equtext
\begin{equation}\label{eq:load}
\begin{aligned}
Q_{\rm Read} = \frac{{B{W_O}{H_O}{C_O}}}{{bxyz}}\left( {{W_K}{H_K}{C_I}z + bx'y'{C_I}} \right)\\
 \approx \frac{{B{W_O}{H_O}{C_O}}}{{bxyz}}\left( {{W_K}{H_K}{C_I}z + \frac{{W_K}{H_K}bxy{C_I}}{R}} \right)
\end{aligned}\equtext
\end{equation}
where the approximation holds if $R\!\approx\! \frac{W_KH_K}{D^2}$, $x'\!\approx\! Dx$, and $y'\!\approx\! Dy$ (when $x\!\!\gg\!\!1$ and $y\!\!\gg\!\!1$, these approximations hold). The DRAM write volume is  $BW_OH_OC_O$, which does not depend on $\{b,z,y,x\}$. If $bxy\!=\!u\!\approx\!Rz$, then\equtext
\begin{equation}\label{eq:dram}
  Q_{\rm DRAM}\!\approx \!\frac{2BW_OH_OC_OW_KH_KC_I}{\sqrt{Ruz}}\!+\!BW_OH_OC_O.\equtext
  \end{equation}
If  $uz\!\approx\! S$ (for minimizing the read volume) and $\frac{W_KH_KC_I}{\sqrt{RS}} $ $\gg \!1$ (for ignoring the write volume), \eqref{eq:dram} satisfies  Theorem~\ref{theo:conv}. %Note that only memory load is considered here, since writing outputs involves a constant volume of off-chip DRAM access and the volume is typically small due to the OutR pattern.
This implies that, to reach the minimum off-chip communication, most of the effective on-chip memory should be assigned to Psums (since $uz\!\approx\! S$). The fundamental principle behind this conclusion is to \textit{use the least inputs to produce the most outputs}, implying that data reuse is maximized. In addition, for layers with few weights, the lower bound of \eqref{eq:conv} may not be tight, since $\frac{W_KH_KC_I}{\sqrt{RS}}\!\gg \!1$ does not hold and the write volume $BW_OH_OC_O$ cannot be ignored.

In fact, a few prior studies have more or less discussed similar dataflows~\cite{select_date2018,one_isca2016,gpuconv_dac2017}. However, they failed to find the superiorities of this dataflow due to the intuitive analysis and the lack of theoretical basis. %and carelessly abandoned it.
Ref.~\cite{one_isca2016} evaluated several OutR dataflows but the poor implementations brought $\sim$50\% of the energy consumed by inter-PE communication which is actually unnecessary. Ref.~\cite{select_date2018} considered an OutR dataflow but the tiling sizes were not properly selected. The convolution implementation proposed in~\cite{gpuconv_dac2017} is for graphics processing units rather than for hardware accelerators. Ref.~\cite{one_arxiv2015} proposed a dataflow for CNN accelerators which explicitly converts any convolutional layer into an MM without exploiting WndR.

\subsection{Workload and Storage Mapping with Minimized On-Chip Communication}

Here we focus on the computation of an iteration (i.e., the red line in Fig.~\ref{fig:dataflow}). The required inputs and weights for an iteration have been loaded to the GBuf. The workload of an iteration is mapped to a PE array that consists of $p\!\times \!q$ PEs. %Before computation, PEs  load the required inputs and weights from  GBuf to  Regs. The produced Psums also reside in  Regs. The on-chip communication including  Gbuf access and Reg access should also be minimized. %Since we are designing the accelerator architecture at the same time and the PE array size is not determined, we introduce a co-design methodology for workload mapping and PE array size to optimize the on-chip communication.
We will introduce a workload and storage mapping scheme to minimize both GBuf communication and Reg communication.

\subsubsection{Minimizing GBuf Communication}\label{sec:gbuf}

There is a major difference between optimizations of the off-chip communication and the on-chip communication. When minimizing the  off-chip communication, since the problem scale can be arbitrary but the hardware resources are fixed, tiling is necessary and the workload is finished by a number of sequential iterations. For an iteration, however, since the output sub-matrix size is limited by the on-chip memory capacity, it is possible to design the PE array size and the Reg capacity such that the hardware resources can  handle the workload of an iteration at a time. %The concurrency of PEs can be explored to reduce on-chip communication. %Second, when optimizing  off-chip communication, the on-chip memory is too small to hold all inputs and all weights. For an iteration, however, the required inputs and weights of an iteration have been loaded to GBuf.
This difference  leads to a different lower bound --- the loaded inputs and weights (in the GBuf) can be read exactly once. This is no doubt the minimum possible GBuf communication.

%From Fig.~\ref{fig:imple} we have an interesting finding --- an iteration is just a shrunken convolutional layer. %If we only look at the red blocks and the output sub-matrix colored green in Fig.~\ref{fig:imple}, they just form a smaller convolutional layer.
%Hence, the same principle for minimizing  off-chip communication may also be adopted to minimize  on-chip communication. There is a major difference.

%There is a major difference that for an iteration, there are multiple PEs operating concurrently which can share data, but when minimizing the off-chip communication, all iterations are computed in sequential.

%we assume that a PE is the smallest computational unit that comprises a single path of $N_{\rm MAC}$ cascaded multiplication-and-accumulation (MAC) units that compute $N_{\rm MAC}$ Psums belonging to the same output in one cycle (if a PE has multiple paths or multiple pipeline stages of MAC units, we can split the PE into multiple smallest PEs).

\begin{figure}[t]
  \centering
  \includegraphics[width=1\columnwidth]{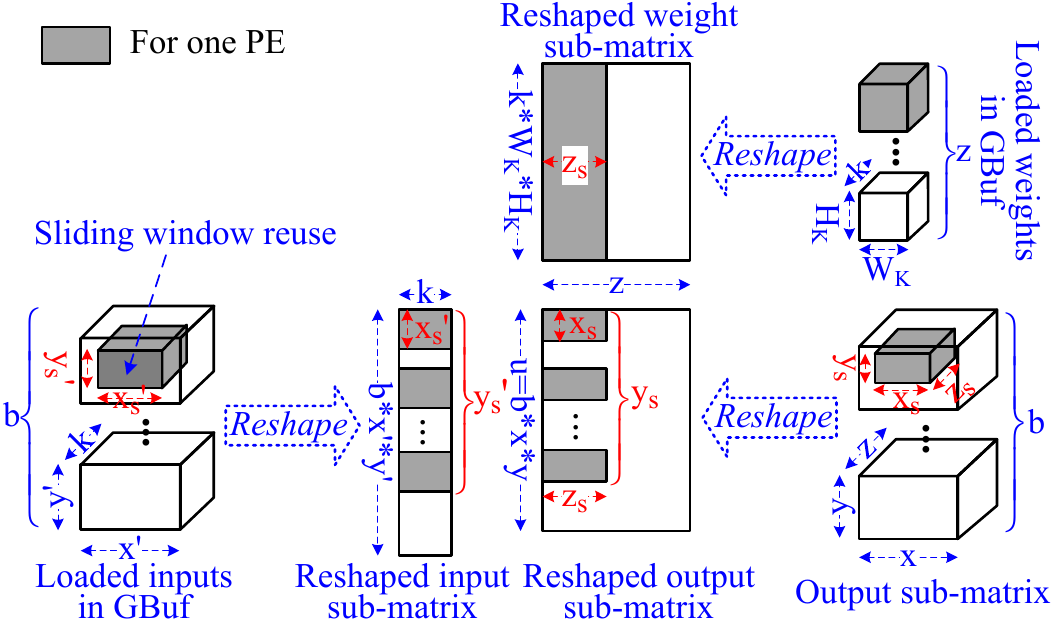}\\
  \graphlabel
  \caption{Workload mapping  in an iteration.}\label{fig:onchip}
\end{figure}

Without loss of generality, a PE is the smallest computational unit that has a multiplication-accumulation (MAC) unit. %Suppose that each PE computes a $z_s\!\times\! y_s\!\times\! x_s$ 3D block in the output sub-matrix,
Like the dataflow to minimize the off-chip communication, each PE computes $x_s\!\times \!y_s$ outputs in $z_s$ output channels, so each PE contributes to a $z_s\!\times\! y_s\!\times\! x_s$ ($\ge\!{(bxyz)}\!/\!{(pq)}$) block in the output sub-matrix, as illustrated in Fig.~\ref{fig:onchip}.
%suppose that each PE computes $x_s\!\times \!y_s$ outputs in $z_s$ output channels, so a PE contributes to $z_s\!\times\! y_s\!\times\! x_s$ ($=\!\frac{bxyz}{pq}$) outputs.
The produced outputs by $p\!\times \!q$ PEs should cover the reshaped output sub-matrix ($bxy\!\times\! z$).  Fig.~\ref{fig:subiter}  details the workload mapping for two PEs (PEs 1,1 and 1,2) and the workload of one PE. For one PE, $x_s' y_s' k$ inputs and $z_skW_KH_K$ weights are needed (remember that $k\!=\!1$ in practice). However, we do not need to load them at a time. To enable WndR on each $x_s'\!\times\! y_s'$ plane (see Fig.~\ref{fig:onchip}), $x_s'\!\times\! y_s'$ inputs (for one PE) are loaded to the Regs. Since WndR cannot be applied to weights, we just load $z_s$ weights (for one PE) to the Regs. In an iteration, if updating all outputs once is called a \textit{pass} (the $i$th pass computes the $i$th Psums of all outputs), in each pass, a PE uses $x_sy_s$ inputs and $z_s$ weights to produces $x_sy_sz_s$ Psums (see the workload of one PE shown in Fig.~\ref{fig:subiter}). A pass needs $x_sy_sz_s$  clock cycles. The loaded inputs can be used for $W_KH_K$ passes (because WndR is exploited in the Regs). The loaded weights can be used just for 1 pass, so a PE needs to load $z_s$ weights to the Regs in every pass.  To complete an iteration, $p\!\times\! q$ PEs need $kW_KH_K$ passes.

\begin{figure}[t]
  \centering
  \includegraphics[width=1\columnwidth]{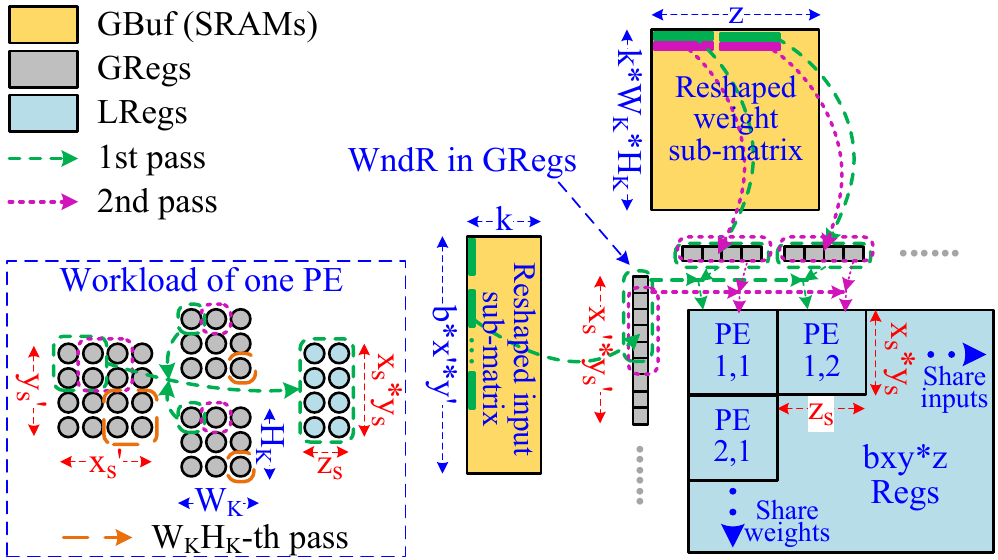}\\
  \graphlabel
  \caption{Workload and storage mapping in an iteration.}\label{fig:subiter}\graphtext
\end{figure}

When considering the PE array, PEs in the same row share the loaded inputs and PEs in the same column  share the loaded weights. As a result,  each weight in the GBuf is read exactly once, reaching the minimum communication. The average read count of each input in the GBuf is ${(x_s'y_s')}\!/\!{(x_sy_s)}$ which is larger than 1. The extra reads are from the halos (i.e., the inputs   out of the $x_s\!\times \!y_s$ rectangle but in the $x_s'\!\times\! y_s'$ rectangle) on each input channel. It is possible to avoid reading extra halos by designing a complicated data transmission network, as an input in a block's halo is also an input of another block, such that each input in the GBuf is also read exactly once. We prefer reading extra halos as it  simplifies the hardware design and regularizes the read patterns. Ideally, the Regs for storing inputs and weights can be global Regs (GRegs) instead of PEs' local Regs (LRegs) (for example, in Fig.~\ref{fig:subiter}, the $x_s'\!\times \!y_s'$ GRegs are shared by the first PE row). In practice, to avoid large fanouts and long latency of long wires, we partition the PE array into groups and each group shares a set of GRegs, with little extra Reg communication.

We choose to store Psums in PEs' LRegs. An alternative way is to store  Psums in the GBuf, which reduces the Reg capacity. However, a Psum needs to  be loaded to a Reg when it is being updated and  stored back to the GBuf when updated, resulting in lots of data shuffling between the GBuf and Regs, and thus, high energy consumption. Hence, storing Psums in the GBuf is not  energy efficient. Keeping Psums in Regs completely avoids GBuf access for Psums. Thus, the  communication between the GBuf and  Regs is minimized. %This strategy also avoids inter-PE communication.

By utilizing our workload and storage mapping, the GBuf capacity can be reduced. Since  weights are read row by row from the reshaped weight sub-matrix and  inputs are read column by column from the reshaped input sub-matrix (see Fig.~\ref{fig:subiter}), we do not need to load $kW_KH_K\!\times\! z$ weights and $bx'y'\!\times\! k$ inputs to the GBuf at a time. Instead, we only need one row of SRAMs for weights and one column of SRAMs for inputs. Once data in the GBuf are loaded to the GRegs, the GBuf is used for prefetching data for the subsequent pass. %This strategy reduces the GBuf capacity. %we can use two rows of global Regs to achieve the same goal. One row is used for MAC computation and the other row is used for  prefetching weights.

\subsubsection{Minimizing Reg Communication}

Psums are stored in PEs' LRegs. Since each MAC operation needs a Reg write, the minimum number of Reg writes is the number of MAC operations, i.e., \equtext%$Q_{\rm Reg}\!=\!{\rm {\#~of~MACs}}$,
%\begin{equation}\label{eq:rfw}
%  Q_{\rm Reg}=\frac{\rm {\#~of~MACs}}{N_{\rm MAC}}.\vspace{-4pt}
%\end{equation}
%Remember that $N_{\rm MAC}$ is the number of cascaded MAC units in a PE. If $N_{\rm MAC}\!=\!1$,  the number of Psum writes equals to the number of MACs,
\begin{equation}\label{eq:reg}
  Q_{\rm Reg}={\rm {\#~of~MACs}}.\equtext
\end{equation}
This is no doubt the minimum Reg communication. Keeping Psums in  LRegs naturally reaches this lower bound,  which minimizes the dynamic energy of LRegs. On the other hand, the static energy of  LRegs should also be optimized.

Suppose that each PE has $r$ ($\ge x_sy_sz_s$) LRegs to store Psums. For a PE, in each cycle, at most one Reg is written and the other $r\!-\!1$ Regs just consume static energy. If $r$ is large, the static energy consumption of the Regs may dominate the total Reg energy. Increasing the PE array size (i.e., $pq$) can reduce $r$, with increased arithmetic component power. %Therefore, there is a tradeoff between the logic power and the register power.
However, with more PEs, the execution time is reduced so that the energy of the arithmetic components almost keeps unchanged.
%This can be explained as follows. Let $P_{\rm MAC}$, $P_{\rm Rd}$, and $P_{\rm Rs}$ be the MAC unit power, the Reg dynamic power, and the Reg static power, respectively. The energy consumption of the PE array  is approximated by\vspace{-4pt}
%\begin{equation}\label{eq:earr}
%\begin{aligned}
%  E_{\rm Array}&\propto pq\left[P_{\rm MAC}+P_{\rm Rd}+P_{\rm Rs}\left(r-1\right)\right]\frac{1}{pq}\\
%  &=\left[P_{\rm MAC}+P_{\rm Rd}+P_{\rm Rs}\left(\frac{bxyz}{pq}-1\right)\right].
%  \end{aligned}\vspace{-4pt}
%\end{equation}
From an energy point of view, using more PEs causes lower static energy consumption of the Regs, though the arithmetic power dissipation will increase.

Using GRegs to share inputs and weights to the  PE array completely avoids inter-PE communication. Duplicating inputs and weights from the GBuf to  GRegs brings little extra Reg communication. Thus, the Reg communication is minimized.

\subsection{Summary}

We summarize the communication lower bound here. The theoretical lower bound of the off-chip communication is defined in \eqref{eq:conv}, while a more practical lower bound is described in \eqref{eq:dram}. The lower bound of the GBuf communication is the off-chip communication of inputs and weights. The lower bound of the Reg communication is defined in \eqref{eq:reg}. %According to the above analysis, our dataflow and workload mapping strategy can practically reach the minimum communication in a three-level memory hierarchy.
There are two key conditions to achieve the lower bound: $bxy\!\approx\! Rz$ (for setting the tiling sizes) and $bxyz\!\approx\! S$ (most of the on-chip memory capacity should be assigned to Psums).

The superiorities of our dataflow and workload and storage mapping scheme come from three aspects. First, our dataflow and workload mapping scheme fully exploit OutR and WndR, and also combine InR and WtR in a balanced way. Our dataflow is actually  a combination of a communication-optimal MM implementation and WndR. The optimal dataflow and workload mapping scheme help reduce both DRAM communication and GBuf communication. Second, the concurrency of PEs is exploited to share inputs and weights by GRegs. Third,  Psums are stored in PEs' LRegs. The last two points both help reduce  GBuf communication and Reg communication. By combining these techniques, our approach can practically reach the minimum communication in a three-level memory hierarchy for convolution accelerations.

\section{Communication-Optimal CNN Accelerator Architecture}\label{sec:arch}
In this section, we  propose a CNN accelerator architecture with minimized communication, based on the theoretical conclusions of the previous section. According to the implication of \eqref{eq:dram}, most of the effective on-chip memory should be assigned to Psums to minimize the off-chip communication. We use an example containing 64KB Psums and $p\!\times\! q\!=\!16\!\times\!16$ PEs to describe the design methodology of our CNN accelerator. We use 16-bit fixed-point arithmetic units, so there are 32K (32768) entries  for Psums and each PE has 128 entries.

\begin{figure}[b]
\graphbottomtext
  \centering
  \includegraphics[width=1\columnwidth]{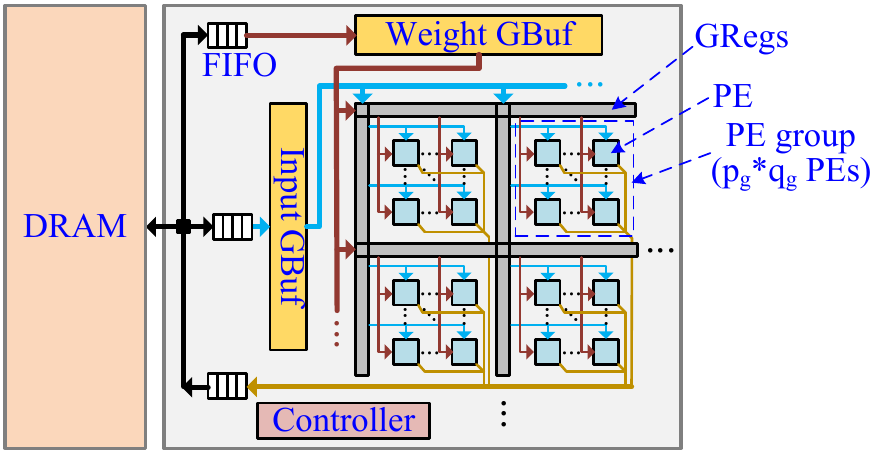}\\
  \graphlabel
  \caption{Architecture of our CNN accelerator.}\label{fig:arch}
\end{figure}

Based on the workload and storage mapping scheme illustrated in Fig.~\ref{fig:subiter}, we design our architecture as shown in Fig.~\ref{fig:arch}. The architecture mainly comprises a PE array, GRegs, two GBufs (an input GBuf (IGBuf) and a weight GBuf (WGBuf)), a controller, and some first-in first-out (FIFO) buffers that connect the off-chip DRAM and the on-chip memories.

\textbf{GBufs:} According to the discussions of Section~\ref{sec:gbuf}, %each PE row  shares the same inputs and each PE column shares the same weights. Ideally, we only need one GReg column and one GReg row to store inputs and weights, respectively. However, long wires of the GReg outputs may cause practical problems (e.g., long latency).
to avoid long wires, the PE array is partitioned into PE groups and each PE group ($p_g\!\times\! q_g$ PEs) shares a set of GRegs (see Fig.~\ref{fig:arch}). In our example, $p_g\!\!=\!\!q_g\!\!=\!\!4$. All GReg rows (columns) store the same weights (inputs), and the same position in all GReg rows (columns) is written at the same time.

We discuss how to determine the sizes of the  GBufs. Remember that most of the effective on-chip memory should be assigned to Psums (i.e., $S\!\approx\!32768$) and the tiling sizes $\{b,z,y,x\}$ should satisfy $bxy\!\approx\! Rz$ to minimize the off-chip communication. If $R\!=\!1$ (i.e., no WndR), $bxy\!\approx\! z\!\approx\!181$. This is the approximate maximum value of $z$, so we set the size of the WGBuf  to 256 entries (0.5KB). With larger $R$, $bxy$  also becomes larger. Considering that the maximum $R$ is typically 9 ($W_K\!=\!H_K\!=\!3$ and $D\!=\!1$, see \eqref{eq:reuse}), the maximum $bxy$ is 543. Since the IGBuf  should store  $bx'y'$ (slightly larger than $bxy$) inputs from $b$ $y'\!\times\! x'$ input channel planes (see Fig.~\ref{fig:onchip}), we set the size of the IGBuf  to 1024 entries (2KB). We leave some extra entries in the GBufs to adapt to various tiling sizes. Even so, the GBuf capacity is still very small. Once data in the GBufs are loaded to the GRegs, the GBufs are used for prefetching  inputs and weights for the subsequent pass. The prefetching is (partially) overlapped with computation.

Inputs and weights stored in the GBufs are just in the order as in the reshaped input and weight sub-matrices (see Fig.~\ref{fig:onchip}). This is the natural order when loading them from  the DRAM. No special order is needed. Inputs are not unfolded so  we can exploit WndR on chip.

\begin{figure}[b]
\graphbottomtext
  \centering
  \includegraphics[width=1\columnwidth]{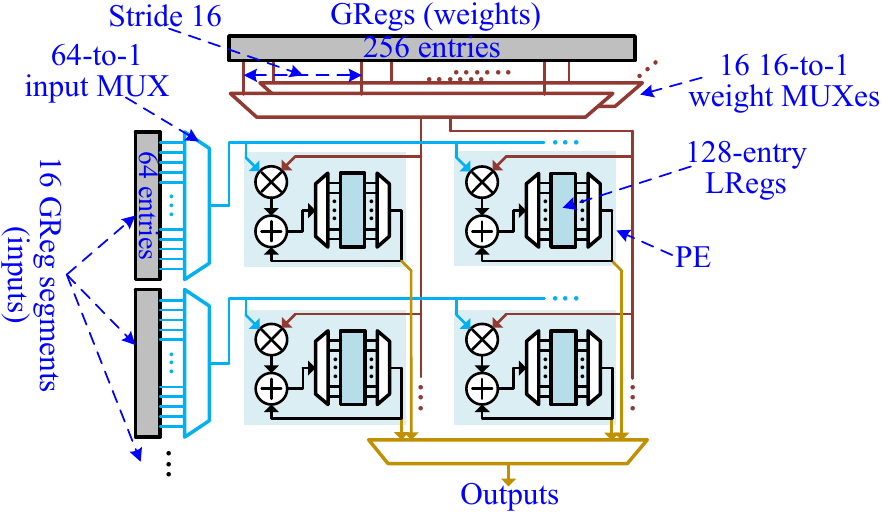}\\
  \graphlabel
  \caption{PE and GReg architectures (numbers are for our example).}\label{fig:pe}
\end{figure}

\textbf{GRegs:} A GReg row (storing weights) is shared by $p_g$ PE rows so $p_g\!\times\! q$ ($4\!\times \!16$) PEs share a GReg row. Data stored in each GReg row are copied from the WGBuf. To adapt to different $z$ values, we elaborate a multiplexer (MUX) structure, as shown in Fig.~\ref{fig:pe}. There are $q$ (16) $\frac{256}{q}$-to-1 (16-to-1) weight MUXes connecting the WGBuf and the $q$ (16) PE columns. Slightly different from the workload mapping shown in Fig.~\ref{fig:subiter}, here the $z_s$ channels computed by a PE is not consecutive but have a stride size $q$ (16). The inputs of the $q$ (16) weight MUXes are arranged in a round-robin way, so that the input range exactly covers all entries of the WGBuf. To adapt to different $z$ and $z_s$ values, we just control the selection signals of the weight MUXes. For instance, if $z\!=\!64$ (so $z_s\!=\!4$), the selection signals of the weight MUXes are from 0 to 3, so that only the first 64 entries of the WGBuf  can be selected. Such a weight MUX structure avoids the use of a complicated data transmission network (e.g., a network-on-chip).

To exploit WndR in the GRegs, each GReg column (storing inputs) is partitioned into $p$ (16) segments. A GReg segment has 64 entries and is shared by $1\!\times\!q_g$ ($1\!\times\! 4$) PEs. Each GReg segment loads $x_s'y_s'$ inputs (see Fig.~\ref{fig:subiter}) from the IGBuf. The  $x_s'y_s'$ inputs can be used in $W_KH_K$ passes to compute $x_sy_sz_s$ Psums. Each GReg segment has a 64-to-1 MUX to provide  inputs to the $1\!\times\!q_g$ ($1\!\times\! 4$) PEs. The selection signals of the input MUXes are from 0 to $x_s'y_s'\!-\!1$ so that only the first $x_s'y_s'$ entries of the GReg segments can be selected.

\textbf{PEs:} A PE comprises a MAC unit and a set of LRegs (128 entries) for Psums. Our architecture does not need LRegs in each PE to store inputs or weights. A PE computes a Psum and writes the accumulated result to an LReg in each cycle. All PEs operate synchronously. This means that, at the same moment, the selection signals of all input MUXes are identical, the selection signals of all weight MUXes are identical, and the read and write positions of all LRegs are also identical.

\textbf{Controller:} Our architecture  has a global controller, which schedules the computational process. It is  a finite-state machine that generates  control signals for all components, including the read/write signals and addresses of all memories and the selection signals of all MUXes. No local controller is needed in each PE.

\section{Experimental Results}\label{sec:result}

\begin{table}[b]
  \centering
  \caption{Five implementations of our architecture.\tablelabel}\label{tab:implem}
   \setlength{\tabcolsep}{5pt}
  \begin{tabular}{|c|c|c|c|c|c|}
  \hline
  {Implementation \#}&{1}&{2}&{3}&{4}&{5}\\
  \hline
  {\# of PEs}&{16$\times$16}&{32$\times$16}&{32$\times$32}&{32$\times$32}&{64$\times$32}\\
  \hline
  {GBuf size (KB)}&{2.5}&{2.5}&{2.5}&{3.625}&{3.625}\\
  \hline
  {LReg size/PE (B)}&{256}&{128}&{64}&{128}&{64}\\
  \hline
  {GReg size (KB)}&{10}&{15}&{18}&{27}&{36}\\
  \hline
  \multirow{2}{*}{\makecell{Effective on-chip\\ memory size (KB)}}&\multirow{2}{*}{66.5}&\multirow{2}{*}{66.5}&\multirow{2}{*}{66.5}&\multirow{2}{*}{131.625}&\multirow{2}{*}{131.625}\\
  {}&{}&{}&{}&{}&{}\\
  \hline
  \end{tabular}
\end{table}

Our CNN accelerator is implemented in Verilog. We synthesize it with Design Compiler based on the 65nm technology. We use Memory Compiler to generate the GBufs. The power dissipation is evaluated with PrimeTime. CACTI~\cite{cacti} is employed to evaluate the latency and energy consumption of a 2GB DDR3 DRAM (the peak bandwidth is 6.4GB/s). The core frequency is 500MHz and the DRAM frequency is 100MHz. A cycle-accurate simulator is built to evaluate the performance with memory access latency taken into account. The representative state-of-the-art, Eyeriss~\cite{eyeriss_jssc2017,one_isca2016}, is the baseline for comparison (detailed off-chip and on-chip communication volumes are reported in~\cite{eyeriss_jssc2017}). The workload is VGGNet-16~\cite{vgg_arxiv2014} with batch size 3, the same as the workload used in~\cite{eyeriss_jssc2017}. VGGNet has diverse layer dimensions, including large/shallow layers, small/deep layers, and layers with medium size/depth.

We evaluate five implementations of our accelerator with different PE numbers and on-chip memory sizes, as listed in Table~\ref{tab:implem}. Table~\ref{tab:energy} lists the  energy consumption of the basic operations, estimated by our simulations.

\begin{table}[t]
  \centering
  \caption{Energy consumption of  operations.\tablelabel}\label{tab:energy}
  \begin{tabular}{|c|c||c|c|}
  \hline
  {MAC}&{4.16pJ}&{LReg (256B) access}&{3.39pJ}\\
  \hline
  {GBuf (0.5KB) access}&{0.30pJ}&{LReg (128B) access}&{1.92pJ}\\
  \hline
   {GBuf (2KB) access}&{1.39pJ}&{LReg (64B) access}&{1.16pJ}\\
  \hline
   {GBuf (3.125KB) access}&{2.36pJ}&{DRAM (2GB) access}&{427.9pJ}\\
  \hline
  \end{tabular}
\end{table}

\subsection{DRAM Access Volume}

\begin{figure}[t]
  \centering
  \includegraphics[width=1\columnwidth]{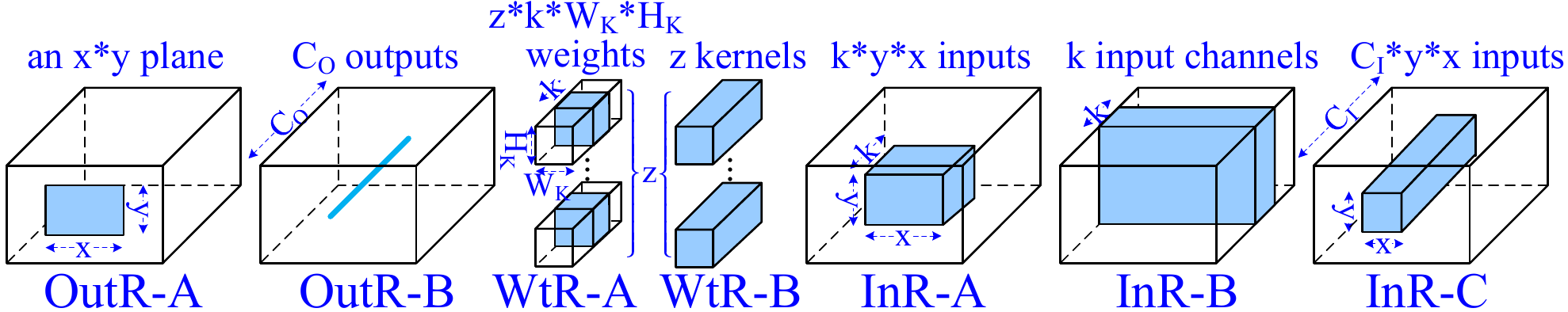}\\
   \graphlabel
  \caption{Different dataflows for comparison.}\label{fig:compare}\graphtext
\end{figure}

\begin{figure}[b]
\graphbottomtext
  \centering
  \includegraphics[width=1\columnwidth]{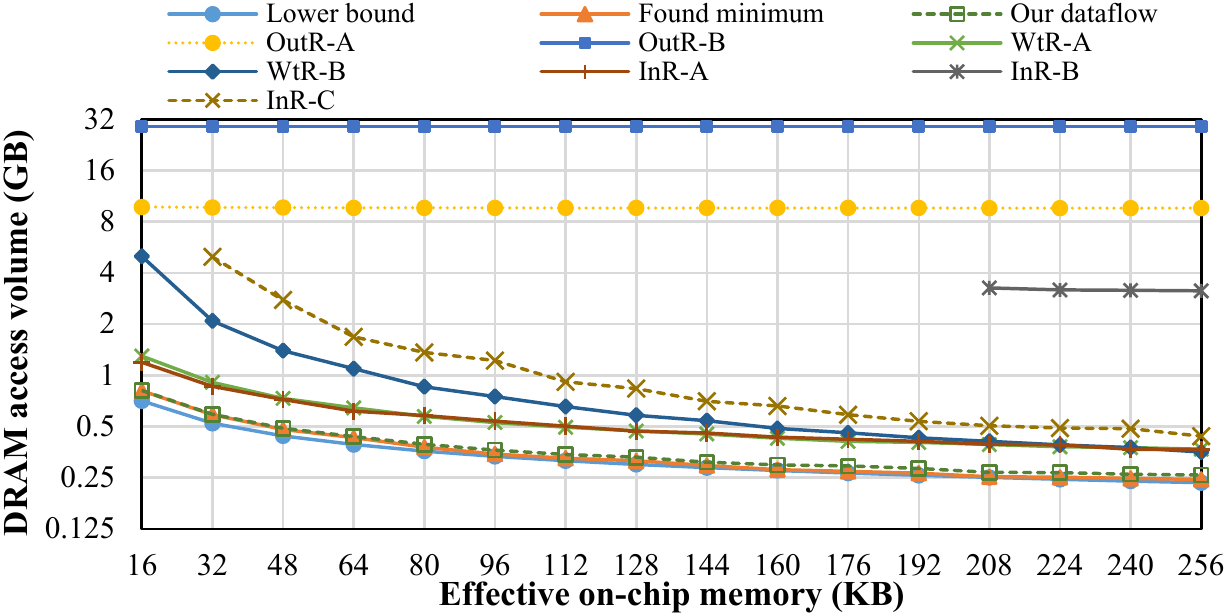}\\
  \graphlabel
  \caption{Comparison of different dataflows under different effective on-chip memory sizes.}\label{fig:sweep}
\end{figure}

We compare our dataflow with other dataflows based on different data reuse patterns, as shown in Fig.~\ref{fig:compare}, in which the colored blocks reside on chip for reuse. For example, in InR-A, a $k\!\times \!y\!\times\! x$ block resides on chip for reuse, while the associated weights and outputs are shuffled on and off chip when necessary. These dataflows should cover the most popular ones used in  literature. For example, ShiDiaoNao~\cite{one_isca2015} uses OutR-A.

Fig.~\ref{fig:sweep} compares the DRAM access volume under different effective on-chip memory sizes. The lower bound is calculated by \eqref{eq:dram}. To make a fair comparison and to remove the impact of improper tiling sizes, the tiling sizes of all dataflows are obtained by exhaustive searches (since the loop order is fixed, searching for the best tiling sizes is fast, typically shorter than 0.1s). The found minimum is obtained by searching for the best dataflow with the best tiling sizes for each layer. Fig.~\ref{fig:sweep} demonstrates that our dataflow produces almost the same DRAM access volume as the found minimum, and the difference is only 4.5\% on average. To understand why our dataflow does not produce the least DRAM access volume for all layers, we have mentioned at the end of Section~\ref{sec:derivation} that the derived lower bound is in the form of $\Omega$ instead of a precise value. However, despite that, it is unnecessary to select the best dataflow from multiple candidates, as the expected improvement in the DRAM access volume is less than 5\%. Our dataflow produces 10\% more DRAM access volume on average than the theoretical lower bound. The 2nd and 3rd best dataflows, InR-A and WtR-A, respectively produce 45.1\% and 45.8\% more DRAM access volume than ours.  %This comparison indicates that our dataflow can produce significantly low DRAM access volume.

\begin{figure}[t]
  \centering
  \includegraphics[width=1\columnwidth]{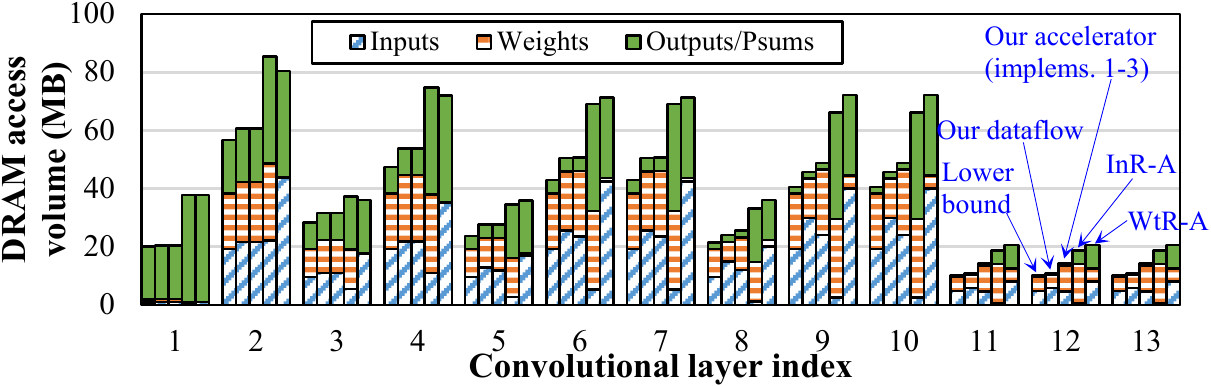}\\
  \graphlabel
  \caption{Per-layer comparison of different dataflows (66.5KB effective on-chip memory).}\label{fig:offc_64kb}\graphtext
\end{figure}

Fig.~\ref{fig:offc_64kb} shows the per-layer DRAM access volume of the lower bound, our dataflow, our implementations 1-3, InR-A, and WtR-A. The difference between our dataflow and our implementation is that the latter has a fixed on-chip memory splitting (e.g., 64KB Psums plus 2.5KB GBufs in our implementations 1-3). Due to this reason, our implementations 1-3 produce 3-4\% more DRAM access than our dataflow, indicating tiny impacts of the fixed on-chip memory splitting. Our dataflow and implementations produce balanced input and weight access volumes, while outputs take up a small portion of the DRAM access volume. For InR-A and WtR-A (the 2nd and 3rd best dataflows), outputs involve a large portion of the DRAM access volume, and the input and weight access volumes are not balanced, leading to much larger memory access volumes.

%%%%%%%%%%%%%%%%%%%%%%%%%%%%%%%%%%%%%%%%%%%%%%%%%%%%%%%%%%

\begin{figure}[b]
\graphbottomtext
  \centering
  \includegraphics[width=1\columnwidth]{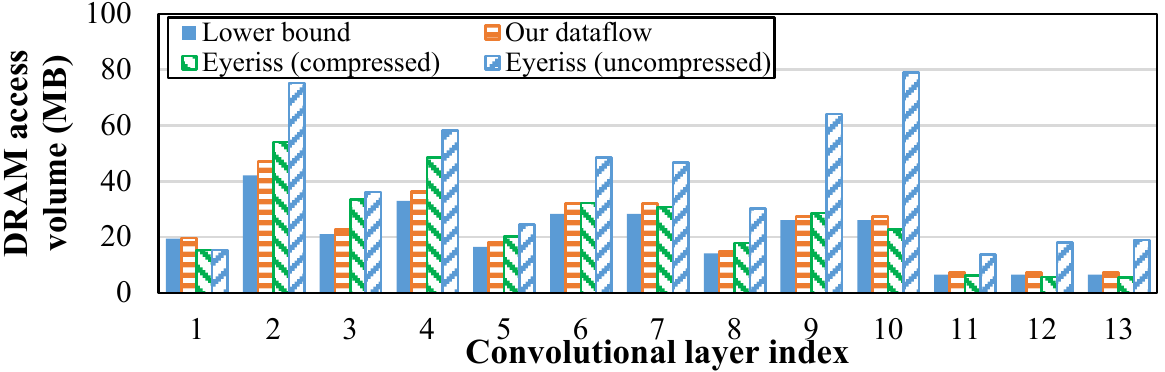}\\
   \graphlabel
  \caption{Comparison  with Eyeriss on DRAM access (173.5KB effective on-chip memory).}\label{fig:offc_eyeriss}
\end{figure}

We try to make an apple-to-apple comparison with published data but find it difficult. Ref.~\cite{eyeriss_jssc2017} reported the DRAM access volume of VGGNet-16 with input compression on Eyeriss. Ref.~\cite{select_date2018} selected the best dataflow with the minimum DRAM access volume from three candidates. Inputs, weights, and outputs are pruned in~\cite{select_date2018}. Our work targets at general CNN accelerators without data pruning/compression, so the results reported in~\cite{eyeriss_jssc2017,select_date2018} are not directly comparable to ours. Instead,  we try to make an approximate comparison.

Eyeriss has a 108KB GBuf but the effective on-chip memory capacity is 173.5KB, since 100KB of the GBuf stores inputs and outputs (the other 8KB is used for  prefetching weights), while  weights are stored in PEs' local SRAMs (each PE has 448B local SRAMs)~\cite{eyeriss_jssc2017}. Under the 173.5KB effective on-chip memory limit, we compare our dataflow and Eyeriss with and without input compression, as shown in Fig.~\ref{fig:offc_eyeriss} and Table~\ref{tab:offc_eyeriss}. Ref.~\cite{eyeriss_jssc2017} has reported the per-layer input compression ratios of VGGNet-16 but the proportion of the input access volume in the total access volume is not reported. We use the proportion of our dataflow to evaluate the off-chip DRAM access volume for Eyeriss without input compression. %For our dataflow with fixed tiling parameters, the tiling is $\{3,128,14,14,1\}$.
%The DRAM access volume of the four cases are 274.8MB, 299.7MB, 321.3MB, and 528.8MB, respectively. Accordingly, the rates of DRAM access per MAC of the four cases are 0.003, 0.0033, 0.0035, and 0.0057, respectively.
Our dataflow reduces 43.3\% DRAM access volume than Eyeriss without input compression. Our dataflow even produces 6.7\% less DRAM access volume than Eyeriss with input compression.

We notice from Fig.~\ref{fig:offc_eyeriss} that for layer 1, Eyeriss produces a lower DRAM access volume than the lower bound. This is because the derived  lower bound is in the form of $\Omega$ instead of a precise value. It represents the asymptotic relation between the off-chip communication volume and the on-chip memory capacity when the problem scale is large enough. Special cases  exist for small workloads. However, the first layer typically takes up a small portion of the off-chip communication volume, so that its impact on the overall energy efficiency and performance is negligible.

\begin{table}[t]
  \centering
  \caption{Comparison  with Eyeriss on DRAM access (173.5KB effective on-chip memory).\tablelabel}\label{tab:offc_eyeriss}
  \begin{tabular}{|c|c|c|}
  \hline
  {}&{DRAM access (MB)}&{DRAM access/MAC}\\
  \hline
  {Lower bound}&{274.8}&{0.0030}\\
  {Our dataflow}&{299.7}&{0.0033}\\
  {Eyeriss (compr.)}&{321.3}&{0.0035}\\
  {Eyeriss (uncompr.)}&{528.8}&{0.0057}\\
  \hline
  \end{tabular}
\end{table}

%The on-chip memory capacity is 108KB (the same as that in Eyeriss). Ref.~\cite{select_date2018} shows that the minimum off-chip communication volume is 220MB while the overall pruning rate is about 46\%, so the equivalent off-chip communication volume without pruning is about 407MB. Ref.~\cite{select_date2018} considers an ideal case that completely avoids accessing zeros which is an unrealistic assumption. We have to compress zeros and access the compressed data stream, so the actual off-chip communication volume should be larger than 407MB. Ref.~\cite{eyeriss_jssc2017} reports 321.1MB off-chip communication with 58.6\% zero inputs pruned, so the equivalent off-chip communication volume without pruning is about 473MB, assuming that inputs  take up $\frac{1}{3}$ of the  off-chip communication. For an intuitive comparison, our dataflow produces 351MB off-chip communication in the same case.

Compared with FlexFlow~\cite{select_hpca2017} with 192KB on-chip memory (64KB GBuf and 512B/PE local storage) which selects the best dataflow from several candidates, the DRAM access/MAC metric of our dataflow (173.5KB effective on-chip memory) is 33\% better (0.0033 versus 0.0049).

\subsection{GBuf Access Volume}

\begin{figure}[b]
\graphbottomtext
  \centering
  \includegraphics[width=1\columnwidth]{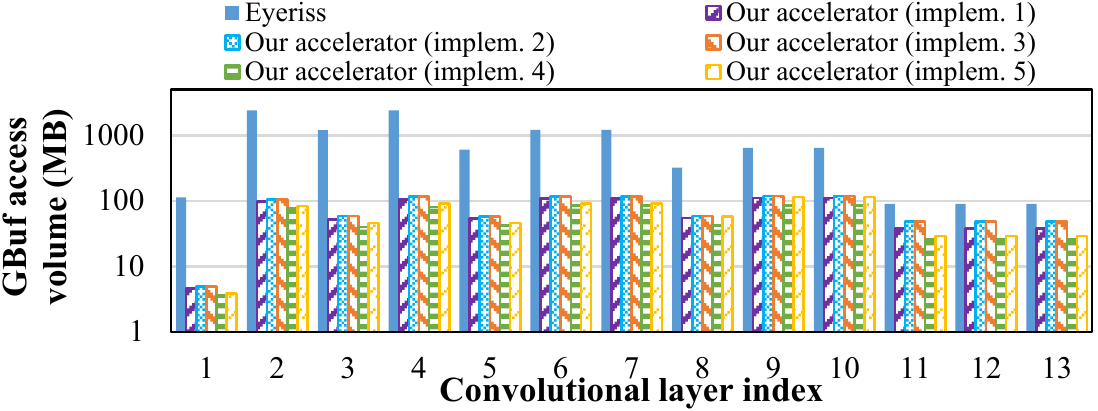}\\
  \graphlabel
  \caption{Comparison with Eyeriss on GBuf access (vertical axis is in logarithmic scale).}\label{fig:gbuf_eyeriss}
\end{figure}

Fig.~\ref{fig:gbuf_eyeriss} shows the GBuf access volume of our accelerator and the comparison with Eyeriss.  Our implementations (with smaller total and effective on-chip memory capacities) produce much less GBuf communication than Eyeriss, and the reduction factors are 10.9-15.8$\times$. The large reduction is  due to the elimination of data shuffling between the GBuf and LRegs.

To understand how our accelerator reaches the minimum GBuf communication, we list the DRAM and GBuf access volumes of implementation 1 in Table~\ref{tab:gbuf}. For weights, the GBuf read and write volumes respectively equal to the DRAM read volume, reaching the theoretical lower bound. For inputs, the GBuf write volume is slightly larger than the DRAM read volume, because the tiling-based dataflow causes some input or output blocks out of the input or output boundaries, resulting in a few redundant GBuf writes. The GBuf read volume for inputs is 1.67$\times$ of the DRAM read volume for inputs. The extra reads are from the halos of convolution inputs, which is explained in Section~\ref{sec:gbuf}. The  GBuf read and write volumes are respectively 1.33$\times$ and 1.07$\times$ of the DRAM read volume, indicating that our accelerator roughly reaches the theoretical lower bound of the GBuf communication.

\begin{table}[t]
  \centering
  \caption{Ratio of GBuf access volume to DRAM access volume (for our accelerator implementation 1).\tablelabel}\label{tab:gbuf}
  \setlength{\tabcolsep}{4pt}
  \begin{tabular}{|c|cc|cc|}
  \hline
  {}&\multicolumn{2}{c|}{DRAM access}&\multicolumn{2}{c|}{GBuf access}\\
  %\cline{2-5}
  {}&{Read}&{Write}&{Read}&{Write}\\
  \hline
  {Inputs}&{187.5MB}&{0}&{313.5MB (1.67$\times$)}&{216.2MB (1.15$\times$)}\\
  {Weights}&{196.6MB}&{0}&{196.6MB (1.00$\times$)}&{196.6MB (1.00$\times$)}\\
  {Outputs}&{0}&{77.5MB}&{0}&{0}\\
  %\hline
  %{Inputs+Weights}&{384.1MB}&{0}&\multicolumn{2}{c|}{845.78 (2.31$\times$)}\\
  \hline
  \end{tabular}
\end{table}

\subsection{Reg Access Volume}

\begin{figure}[t]
  \centering
  \includegraphics[width=1\columnwidth]{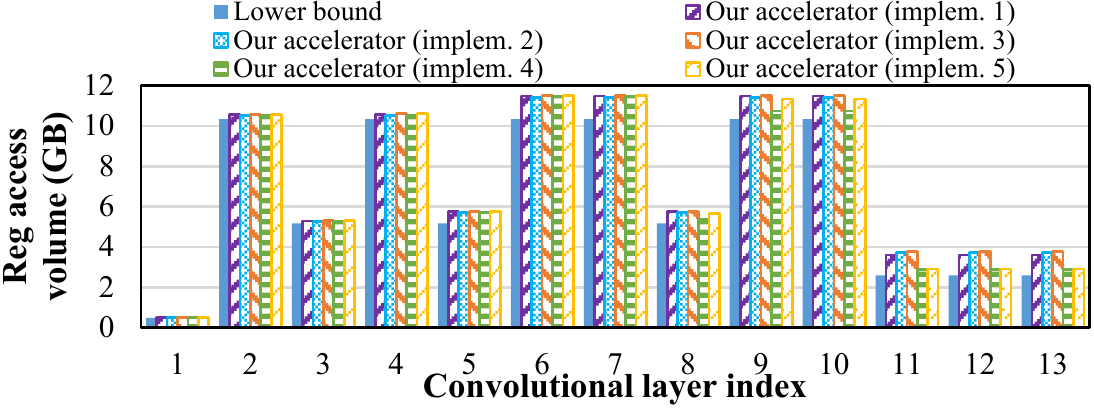}\\
  \graphlabel
  \caption{Reg access volume of our accelerator.}\label{fig:regc}\graphtext
\end{figure}

Fig.~\ref{fig:regc} shows the Reg access volume of our accelerator and the comparison with the lower bound. The lower bound is calculated from~\eqref{eq:reg}. The Reg access volume of our accelerator is only 5.9-11.8\% larger than the lower bound, indicating that our accelerator almost reaches the theoretical lower bound of the Reg communication. The extra Reg communication is from a) the GReg communication, and b) Psums that are out of the output boundary caused by the tiling-based approach.

We are not able to make a numerical comparison with any existing CNN accelerator on the Reg communication since  no similar result was found. For an intuitional comparison with Eyeriss (and other accelerators which propagate data in the PE array, e.g., \cite{one2_tcasi2018,one_icpads2017}), our architecture is expected to  reduce the Reg communication severalfold, because Eyeriss not only writes Psums to  Regs in each cycle (which our accelerator also has), but also propagates inputs, weights, and Psums in the PE array (which our accelerator does not have).

\subsection{Energy Efficiency and Performance}

\begin{figure}[t]
  \centering
  \includegraphics[width=1\columnwidth]{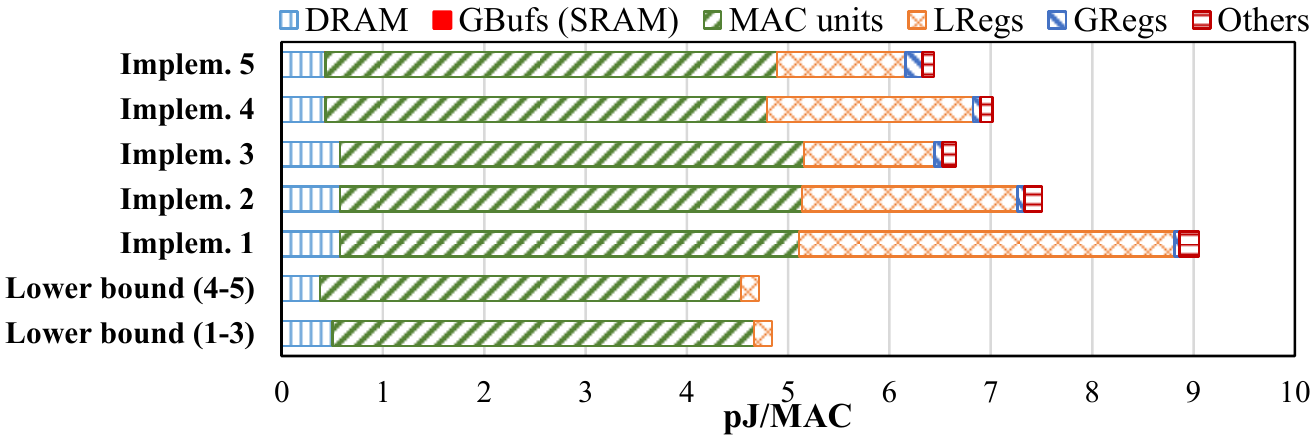}\\
   \graphlabel
  \caption{Energy efficiency of our accelerator.}\label{fig:energy}\graphtext
\end{figure}

Fig.~\ref{fig:energy} shows the energy efficiency (in pJ/MAC) of our accelerator and the comparison with the lower bound. The lower bound is calculated by adding together the DRAM access energy (under the corresponding effective on-chip memory capacity limit), the MAC energy, and the Reg write energy (of (\# of MACs) writes). The lower bound describes the essential energy consumption to complete the MAC operations. MAC operations and Regs dominate the energy consumption of our accelerator. Our accelerator almost reaches the lower bound for DRAM communication and MAC operations. For the Reg energy, our accelerator brings higher energy than the  lower bound. The extra Reg energy is mainly due to the static energy consumption of the LRegs. With fewer LRegs in each PE, the Reg energy consumption is decreased.  Even so, MAC operations take up the largest portion in the total energy consumption, implying  that our accelerator is computation dominant. The gap between the energy efficiency of our implementations and the best value is only 37-87\%, indicating that our accelerator  roughly reaches the  best energy efficiency.

According to the measured data reported in~\cite{eyeriss_jssc2017}, the energy efficiency of Eyeriss with input compression and zero gating is 22.1pJ/MAC (for on-chip aspects). As a direct numeric comparison, our accelerator (by simulations) without data compression or gating is 2.61-3.68$\times$ more energy efficient than Eyeriss for  on-chip aspects. %\footnote{Note that Eyeriss compresses inputs and gates data if an input is 0 which reduce the energy consumption, but our accelerator does not have these features.}.

\begin{figure}[t]
  \centering
  \includegraphics[width=1\columnwidth]{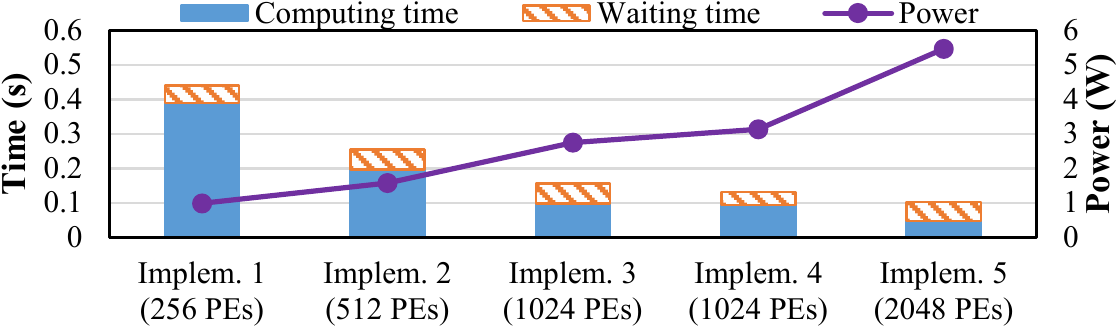}\\
  \graphlabel
  \caption{Power dissipation and performance of our accelerator.}\label{fig:perfpow}\graphgraph
\end{figure}

Fig.~\ref{fig:perfpow} shows the performance and power dissipation of our accelerator. With more PEs, the execution time is reduced and the power is increased. The proportion of waiting time increases with more PEs. With reduced computational time, the memory access latency cannot be fully overlapped by computation so it affects  the  execution time. Compared with Eyeriss, our five implementations achieve 9.8-42.3$\times$ performance gain, with memory access latency taken into account.

\subsection{Memory and PE Utilizations}

\begin{figure}[t]
  \centering
  \includegraphics[width=1\columnwidth]{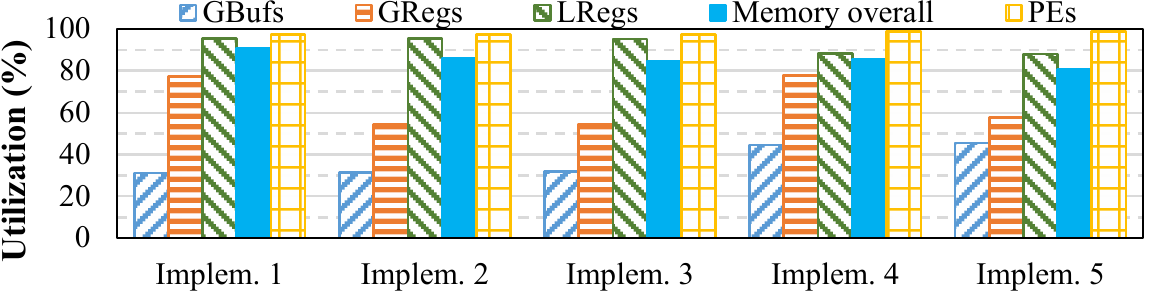}\\
  \graphlabel
  \caption{Memory and PE utilizations (average utilization of all layers) of our accelerator.}\label{fig:util}\graphtext
\end{figure}

Fig.~\ref{fig:util} shows the memory and PE utilizations of our accelerator. The GBuf and GReg utilizations are low because we have some redundant SRAMs and GRegs to adapt to   diverse tiling sizes caused by different layer dimensions. The LReg utilization keeps high ($>$88\%) in different implementations, indicating that most of the LRegs are utilized. Increasing the PE number can lower the LReg utilization, due to the small workload of each PE. Since the LRegs dominate the on-chip memories, the overall memory utilization is also high (80.6-91.0\%). The PE utilization keeps very high ($>$97\%). In fact, all PEs are busy in our implementations. The small quantity of useless PE workload is caused by the tiling-based approach.

\section{Conclusions}\label{sec:concl}

In current CNN accelerators, communication dominates the energy consumption and consumes much more energy than computation. In this work, we provide the theoretical lower bounds of both off-chip communication and on-chip communication. Based on the theoretical results, we elaborate our communication-optimal dataflow as well as a communication-optimal accelerator architecture. We  demonstrate by both theoretical analysis and experimental results that our dataflow and architecture are able to practically reach the minimum communication in a three-level memory hierarchy. Our CNN accelerator is computation dominant and the energy efficiency is close to the theoretical best value.

\section*{Acknowledgements}

This work was supported in part by  National Key R\&D Program of China under Grant 2018YFA0701500,  in part  by Key Research Program of Frontier Sciences, CAS, under Grant ZDBS-LY-JSC012, in part by  National Natural Science Foundation of China under Grants 61804155 \& 61834006, in part by    Youth Innovation Promotion Association CAS,   in part by  Young Elite Scientists Sponsorship Program by CAST under Grant 2018QNRC001, and  in part by  Beijing Academy of Artificial Intelligence (BAAI). % under Grant BAAI2019QN0402.
We thank Profs. Mingzhe Zhang and  Zidong Du for their great help in answering the reviewers' comments.

%\bibliographystyle{IEEEtran}
%\bibliography{ref}


\end{document}